\documentclass[12pt]{amsart}
\usepackage[margin=2cm]{geometry}
\usepackage{amsfonts,mathrsfs}
\usepackage{amssymb}
\usepackage[dvips]{graphics}
\usepackage{epsfig}
\usepackage{euscript}
\usepackage{color}
\usepackage[colorlinks,pdfstartview=FitB]{hyperref}
\hypersetup{allcolors=blue}
\usepackage{caption}
\usepackage{subcaption}

\newtheorem{thm}{Theorem}[section]
\newtheorem{cor}[thm]{Corollary}
\newtheorem{lem}[thm]{Lemma}
\newtheorem{conj}[thm]{Conjecture}
\newtheorem{cl}[thm]{Claim}
\newtheorem{rem}[thm]{Remark}
\newtheorem{ex}[thm]{Example}

\theoremstyle{definition}

\numberwithin{thm}{section}

\newcommand{\R}{{\mathord{\mathbb R}}}
\newcommand{\N}{{\mathord{\mathbb N}}}

\newcommand{\s}{\mathcal{S}}

\newcommand{\C}{{\mathord{\mathbb C}}}
\newcommand{\Z}{{\mathord{\mathbb Z}}}

\def\idty{{\mathchoice {\mathrm{1\mskip-4mu l}} {\mathrm{1\mskip-4mu l}} %
{\mathrm{1\mskip-4.5mu l}} {\mathrm{1\mskip-5mu l}}}}

\DeclareMathOperator{\Span}{span}

\DeclareMathOperator{\diag}{diag}
\begin{document}

\title[Quantum Walks in a Periodic Local Field]{Exponentially decaying velocity bounds of quantum walks in periodic fields}

\author[H. Abdul-Rahman]{Houssam Abdul-Rahman}
\address{Houssam Abdul-Rahman\\ Mathematics, Science Division, NYU Abu Dhabi.}
\email{ha2271@nyu.edu}
\author[G. Stolz]{G\"unter Stolz}
\address{G\"unter Stolz\\ Department of Mathematics\\
University of Alabama at Birmingham\\
Birmingham, AL 35294 USA\\ Phone +1-205-934-2154}
\email{stolz@uab.edu}

\date{\today}


\begin{abstract}
We consider a class of discrete-time one-dimensional quantum walks, associated with CMV unitary matrices, in the presence of a local field. This class is parametrized by a transmission parameter $t\in[0,1]$. We show that for a certain range for $t$, the corresponding asymptotic velocity  can be made arbitrarily small by introducing a periodic local field with a sufficiently large period. In particular, we prove an upper bound for the velocity of the $n$-periodic quantum walk that is decaying exponentially in the period length $n$. Hence, localization-like effects are observed even after a long number of quantum walk steps when $n$ is large.
\end{abstract}

\maketitle


\tableofcontents

%
%

\allowdisplaybreaks

\section{Introduction}

Motivated by the widespread applications of classical random walks in all branches of science from search algorithms to the study of stock markets, quantum walks, the quantum analogue of random walks,  have received great  interest in different disciplines. In many settings they have been proven to be superior over their classical counterparts. 
For example, quantum walks are shown to give exponential speedups  over many classical algorithms, see e.g.,  \cite{Alg-1, Alg-2, Alg-3}.
Quantum walks represent now key operations in quantum computing, information, simulation, and communication. 
They have been modeled to implement universal quantum gates for quantum computers \cite{App-QGates1, App-QGates2}, and quantum simulations \cite{App-QS1,App-QS2}. Quantum walks also represent relatively simple models to study and understand quantum dynamics, see e.g. \cite{QD-1, BB1988, QD-2, QD-3} in the physics literature, and \cite{QDM-1, BHJ, QDM-2, HJS9} for their mathematical analysis.
Moreover,  recent highly controllable experimental realizations of quantum walks with photons are emerging with  promises for advances in dynamical quantum walk simulations and in different quantum technologies, e.g., \cite{Exp-1, Exp-2, Exp-3, Exp-4, Exp-5}. More recently, quantum walks are used in detecting topological order in condensed matter physics \cite{TopOrd-1, TopOrd-2,TopOrd3, TopOrd4}.

Among the several types of quantum walks that have been introduced and studied  in different contexts, we consider in this work a class of discrete time quantum walks on $\Z$ in the presence of a local field. Consider a quantum walker with two internal degrees of freedom (e.g., spin states) sitting on $\Z$. The unit time step dynamics is consisting of the application of three consecutive processes by the means of unitary transformations: An update of the internal degree, followed by a one step shift on $\Z$ conditioned on the internal state, then a certain phase is picked up by the walker. The first two steps, the update of inner degree and the conditional shift, correspond to a unitary operator $U_t$ with parameter $t\in[0,1]$ that we will refer to as the \emph{transmission parameter}. It determines the velocity of the walker in the range from no transport when $t=0$ to the trivial ballistic transport when $t=1$. The last step of the quantum walk that applies phases corresponds generally to a unitary diagonal operator $\mathcal{D} = \diag(e^{i\omega_j})$, $j\in \Z$, $\omega_j \in \R$, thought of as a ``local field'', and for our main result, this local field is taken to be periodic. Thus, generally speaking,  our model takes the form
\begin{equation}\label{GeneralModel}
U_{t,\mathcal{D}}=\mathcal{D}U_t.
\end{equation}
In this work, the field-free quantum walk $U_t$ has the same five diagonal structure as the so-called CMV-matrices (see e.g., \cite{CMV1, CMV-Simon}), see \eqref{U-CMV} below, which have found much interest in the theory of orthogonal polynomials on the unit circle, where they arise as unitary analogues of Jacobi matrices. The spectral properties of different CMV matrices have been investigated in the last decade, see e.g.,  \cite{CMV-Math1,CMV-Math2, CMV-Math3}.  Models of the form (\ref{GeneralModel}) were first used in the physics literature, e.g., \cite{BB1988}, and their mathematical study was initiated in \cite{BHJ}. 

As is well known, translation invariance of quantum walks on $\Z$  typically  yields large transport, see e.g., \cite{QWLectures9}, \cite[Thm 4]{QWVelocity}, or \cite[Thm 9.1]{QWVelocity-Damanic}. This is the case for our field-free quantum walk $U_t$ for $t\in(0,1]$, see Section \ref{sec:model}. In the disordered case, when the phases determined by  the local field $\mathcal{D}$ are chosen to be i.i.d random variables $\theta_k^{\omega}$, models of the form (\ref{GeneralModel}) are thought of as analogues of the Anderson model, where $U_t$ and $\mathcal{D}_\omega=\diag\{\theta_k^{\omega}\}$ play the roles of the discrete Laplacian and the random potential, respectively. For that, they are referred to as \emph{unitary Anderson models}, for which localization results are proven in different regimes, see, e.g., \cite{J4, J5, HJS6, HS7, HJS9, JM10, HJ14,QW-Loc1}. The Anderson localization phenomena of quantum walks started to receive greater interest in the physical community after their realization of photons in photonics lattices  \cite{Phy-Localization1,Phy-Localization2}.  This is due to  the fact that the generation and manipulation of photons is controllable, and unlike many other quantum technologies,  many photonics systems can be operated ideally in room temperature conditions. For that, \emph{localized quantum walks} are now of potential use for secure transmission of quantum information \cite{App-Info-Transport} and secure quantum memory \cite{App-secureQM}.

Very recent implementations of quantum walks in the physics literature  are stressing that localization-like effects for quantum walks can be observed also in disorder free regimes, e.g., \cite{Phs-PQW1,Phs-PQW2,Phs-PQW3,Phs-PQW4, Phs-QW5}. For example, quasi-periodic arrays of waveguides  are proposed to realize localized quantum walks in \cite{Phs-QW5}, and they have been demonstrated recently in a type of optical fibers with rings of cores constructed with quasi-periodic Fibonacci sequences in \cite{Phs-PQW2}.

In the periodic case, it is generally expected and in some cases rigorously known that periodic quantum systems exhibit non-trivial ballistic quantum transport \cite{AK, DLY}. For example, \cite{DLY} shows the following: If $J$ is a periodic block Jacobi matrix and $X$ the position operator, then $\frac{1}{t}e^{iJt}  X e^{-iJt} \psi \to Q\psi$ as $t\to \infty$, where $\psi \in D(X)$ and $Q$ a bounded selfadjoint  operator with $\text{ker}(Q) =0$, interpreted as the velocity operator. In \cite{AK} a similar result for continuous periodic Schr\"odinger operators in any dimension is proven. Here, we ask whether it is possible to choose a periodic local field (the analogue of periodic potential) such that the quantum walk (\ref{GeneralModel}) has arbitrarily small velocity. It is intuitively clear and not difficult to see, see Lemma \ref{lem:v<t}, that this happens if one chooses the transmission parameter $t$ (the off-diagonal entries) in $U_t$ arbitrarily small. What happens when varying the local field $\mathcal{D}$, in the periodic sense, and keeping the transmission parameter $t$ fixed is less obvious and it is our main interest here. 

We show in our main result Theorem~\ref{thm:main} that for any fixed transmission parameter $t\in(0,1/4)$, the velocity of the quantum walk $U_{t,\mathcal{D}}$ can be made arbitrarily small by choosing the  $n$-periodic field to be $\mathcal{D}={\mathcal D}_n = \diag(e^{2\pi i j/n}), j\in\Z$ and $n$ is sufficiently large.  In particular, we prove upper bounds for the quantum velocity  that are decaying \emph{exponentially} in the period $n$ of the local field $\mathcal{D}_n$ when the transmission parameter is in the range $t\in(0,1/4)$. That is, localization-like effects can be observed, even after a long time,  with periodic fields of large periods.
Moreover, our simulations of dynamics of this quantum walk in the periodic local field $\mathcal{D}_n$ led us to Conjecture \ref{conj}, that the velocity can be made arbitrarily small, in the presence of a periodic local field, for all fixed transmission parameter $t$ in the range $t\in(0,1)$.

In Section \ref{sec:model} we present the general setting of our class of quantum walks in the presence of local fields, and we highlight some basic properties and behaviors of the model with and without local fields. In Section \ref{sec:main-results} we show that the asymptotic quantum velocity cannot be more than $t$ in any local field, see Lemma \ref{lem:v<t}, then we define and discuss a class of periodic local fields in which the velocity of quantum walks can be made arbitrarily small, see  Theorem \ref{thm:main}.  The proof of Theorem \ref{thm:main} consists of two steps: We prove the result for a subsequence of average velocities, and then we use an interpolation argument with the linear  bound in Lemma \ref{lem:v<t}. The first step is the essential part and its proof spans the remaining sections as follows: The expected position of the walker after $N$ quantum walk steps is characterized by $(U_{t,\mathcal{D}})^N$,  this is our central object of interest to study the asymptotic velocity. Section \ref{sec:subsequence} explains how $(U_{t,\mathcal{D}})^N$ can be written as a sum of some banded operators of increasing bandwidths, written in terms of the so called  \emph{elementary symmetric polynomials} (see e.g., \cite{ESP-1,ESP-2}) of non-commuting operators. They are characterized by two parameters $0\leq m\leq \ell\leq n$, such that $\ell+m$ is the maximum possible bandwidth of the associated elementary symmetric polynomials of operators. Theorem \ref{thm:main-step} provides the crucial part in the proof and it shows that if the local field is taken to be the periodic field $\mathcal{D}=\mathcal{D}_n$ defined in (\ref{def:Dn}) below,  all elementary symmetric polynomials of operators that are associated with a bandwidth $<n$ collapse to constant multiples of the identity. This 
is the result of destructive interferences from the $n$-periodic field $\mathcal{D}_n$. Section \ref{sec:Induction} presents a proof of Theorem \ref{thm:main-step} by induction on $\ell$ and $m$.

\section{General setup}\label{sec:model}

We consider a class of quantum walks on $\Z$ of a walker with two internal degrees of freedom, e.g., a spin-1/2, in the presence of a local field. So the Hilbert space is
\begin{equation}\label{HilbertSpace}
\mathcal{H}=\C^2\otimes \ell^2(\Z) \cong  \ell^2(\Z;\C^2),
\end{equation}
where $\C^2=\Span\{|\uparrow\rangle,|\downarrow\rangle\}$ is called the the coin (or spin) space. Here we use the notations
\begin{equation}
|\uparrow\rangle:=\begin{pmatrix}1\\ 0\end{pmatrix} \text{ and }
|\downarrow\rangle:=\begin{pmatrix}0\\ 1\end{pmatrix},
\end{equation}
and we use $\{|n\rangle;\, n\in\Z\}$ to denote the canonical basis of  $\ell^2(\Z)$. We order the basis of $\mathcal{H}$ by defining 
\begin{equation}
\delta_{2j}:=|\uparrow\rangle\otimes|j\rangle \text{ and } \delta_{2j+1}:=|\downarrow\rangle\otimes|j\rangle, \text{ for all }j\in\Z. 
\end{equation}

The one step motion consists of the following consecutive operations. An update of the internal spin is applied via a unitary transformation known as the \emph{coin} operator.  Next, a shift on $\Z$ conditioned on the internal degree state of the walker is performed  using a unitary \emph{conditional shift} (or just \emph{shift}) operator.  Then, a \emph{local field} by means of a diagonal unitary transformation takes place. More precisely, we consider the family of quantum walks on $\mathcal{H}$, parametrized by $t\in[0,1]$, whose one step motion is given by the unitary
\begin{equation}\label{U}
U_{t,\mathcal{D}}=\mathcal{D}U_t, \text{ where }U_t=V_t W_t.
\end{equation}
Here the coin operator $W_t$ and the shift operator $V_t$ for a fixed $t\in[0,1]$ are given as 
\begin{equation}\label{def:WV}
W_t:=\bigoplus_{j\in\Z}
\begin{pmatrix} r & t\\ -t & r\end{pmatrix}_{2j,2j+1} \text{ and }
V_t:=\bigoplus_{j\in\Z}
\begin{pmatrix} r & t\\ -t & r\end{pmatrix}_{2j-1,2j},
\end{equation}
where $r=\sqrt{1-t^2}$. The subscripts $2j,2j+1$ in $W_t$ and $2j-1,2j$ in $V_t$ in (\ref{def:WV}) refer to the decompositions $\bigoplus_{j}\Span\{\delta_{2j},\delta_{2j+1}\}$ and $\bigoplus_{j}\Span\{\delta_{2j-1},\delta_{2j}\}$ of $\mathcal{H}$, respectively. That is, $W_t$ is a direct sum of identical $2\times 2$ orthogonal matrices starting at even indices, and $V_t$ is block diagonal with identical blocks starting at odd indices.

$\mathcal{D}$ in (\ref{U}) represents a generic local field, and it is a diagonal unitary operator   
\begin{equation}\label{def:D}
\mathcal{D} \delta_{j}=e^{i\theta_j}\delta_{j},
\end{equation}
 with phases $\theta_j\in\mathbb{T}=\mathbb{R}\setminus 2\pi\Z$, for $j\in\Z$. Our main result, Theorem \ref{thm:main} assumes a specific class of periodic local fields of the form (\ref{def:D}), it is defined in (\ref{def:Dn}) below. This model originates from the physics literature, see e.g., \cite{BB1988}, then it was investigated mathematically, see e.g., \cite{HJS9, JM10,HJ11, HJ14, J11}, mostly to answer questions about (dynamical) localization when $\theta_j$'s in (\ref{def:D}) are taken to be i.i.d random variables. 
 
Our model corresponds to the consecutive application of the unitary operator $U_{t,\mathcal{D}}$, in (\ref{U}).
Here we have the following remarks:
\begin{itemize}
\item The field-free one step quantum walk $U_t$ in $(\ref{U})$ has a five-diagonal structure. It is the CMV-type 2-periodic unitary operator. They represent unitary analogs of Jacobi matrices, see e.g., \cite{BHJ,CMV1,CMV-Simon}.
\begin{equation} \label{U-CMV}
U_t=V_t W_t=\begin{pmatrix}
 \ddots&&&&&&\\
          & r^2& rt & t^2&&&&&\\
           & -rt& r^2& rt&&&&&\\
           && -rt&r^2&rt&t^2&&\\
           && t^2&-rt&r^2&rt& &\\
           && &&-rt&r^2& rt&  \\
           && &&t^2&-rt&r^2 &\\
           &&&&&&&\ddots
\end{pmatrix}.
\end{equation}
\item The coin operator $W_t$ in (\ref{def:WV}) can be written as 
\begin{equation}
W_t=\begin{pmatrix} r & t\\ -t & r\end{pmatrix}\otimes \idty_\Z.
\end{equation}
It acts nontrivially only on the spin component of $\mathcal{H}$ as follows
\begin{eqnarray}\label{W-act}
W_t&:& |\uparrow\rangle\otimes|j\rangle \mapsto \left(r|\uparrow\rangle-t|\downarrow\rangle\right)\otimes |j\rangle \nonumber\\
&& |\downarrow\rangle\otimes|j\rangle \mapsto \left(r|\downarrow\rangle+t|\uparrow\rangle\right)\otimes |j\rangle,
\end{eqnarray}
for any $j\in\Z$. That is, measurements after the application of $W_t$ yield a flip in the inner degree with (quantum) probability $t^2$ with no shift.

The shift operator $V_t$ defined in (\ref{def:WV}) can be written as
\begin{equation}
V_t=r\idty_{\mathcal{H}}+t|\downarrow\rangle\langle\uparrow|\otimes T_{-1}
-t|\uparrow\rangle\langle\downarrow|\otimes T_1 \text{ where }T_{\pm 1}:=\sum_{j\in\Z} |j\pm 1\rangle\langle j|.
\end{equation}
Note that $T_{\pm 1}$ is the shift right/left operator on $\ell^2(\Z)$, i.e., $T_{\pm1}|j\rangle=|j\pm1\rangle$. $V_t$ shifts (extends) the walker left or right depending on the walker's internal degree. In particular
\begin{eqnarray}\label{V-act}
V_t&:& |\uparrow\rangle \otimes | j\rangle \mapsto r|\uparrow\rangle \otimes |j\rangle+ t|\downarrow\rangle \otimes |j - 1\rangle \nonumber\\
&& |\downarrow\rangle \otimes | j\rangle \mapsto r|\downarrow\rangle \otimes |j\rangle- t|\uparrow\rangle \otimes |j + 1\rangle
\end{eqnarray}
for any $j\in\Z$. i.e., $V_t$ conditionally shifts the walker one step with probability $t^2$ (and it flips the internal degree). Thus, (\ref{W-act}) and (\ref{V-act}) give, for example
\begin{equation}\label{U-act}
U_t\ :\ |\uparrow\rangle\otimes |j\rangle\  \mapsto\ rt |\downarrow\rangle\otimes |j-1\rangle+\left(r^2 |\uparrow\rangle-rt|\downarrow\rangle\right)\otimes | j \rangle+t^2 |\uparrow\rangle\otimes |j+1\rangle.
\end{equation}

\item The boundary values for $t\in[0,1]$ yield the two extreme cases: $t=0$ corresponds to the trivial no-transport case $U_{t=0}=\idty_{\mathcal{H}}$, and $t=1$ gives the naive ballistic transport with a constant velocity of one unit on $\Z$ per a quantum walk step.
\begin{eqnarray}\label{U:t=1}
U_{t=1}&:& |\uparrow\rangle \otimes |j\rangle\mapsto |\uparrow\rangle\otimes |j+1\rangle \nonumber\\
&& |\downarrow\rangle \otimes |j\rangle\mapsto |\downarrow\rangle\otimes |j-1\rangle.
\end{eqnarray}
As is well known, when the coin operator is fixed along the quantum walk, the corresponding transport is typically ballistic  due to translation invariance, see e.g.,  \cite{QWLectures9}, \cite[Thm 4]{QWVelocity}, or \cite[Thm 9.1]{QWVelocity-Damanic}. This is the case for our model in the absence of local fields, in fact, we expect a large transport when  $t\in(0,1)$, and we show that the quantum walk has asymptotic velocity that is at most $t$, see Lemma \ref{lem:v<t}. For this, we refer to $t$ as the \emph{transmission parameter}. 

Figure \ref{fig:QW-different-t} shows the probability distribution of the field-free quantum walk generated by $U_t$ applied iteratively to $|\uparrow\rangle\otimes|0\rangle$ a $100$ times (quantum walk steps) for three values of the transmission parameter $t=0.8$, $t= 0.5$, and $t=0.2$. Similar behavior is observed when a higher number of quantum walk steps are performed. Note that $U_t^N (|\uparrow\rangle\otimes|0\rangle)$ is supported on $[-2N,2N]$ and hence the walker cannot escape the window $[-N,N]$ on $\Z$. Figure \ref{fig:QW-different-t} and similar numerics with higher number of quantum walk steps are showing that the walker is essentially no more than $Nt$ steps from the starting point. It is also observed that the highest probability corresponds to the position of about $Nt$. Hence, the velocity seems to be of order $t$, this is partially demonstrated in Lemma \ref{lem:v<t} below.
Due to translation invariance, the same probability distributions results if the staring point is chosen to be $|\uparrow\rangle\otimes|j\rangle$ where $j\in\Z$. 


\end{itemize}

\begin{figure}
\centering
\begin{subfigure}{.33\textwidth}
  \centering
     \includegraphics[width=1.1\linewidth]{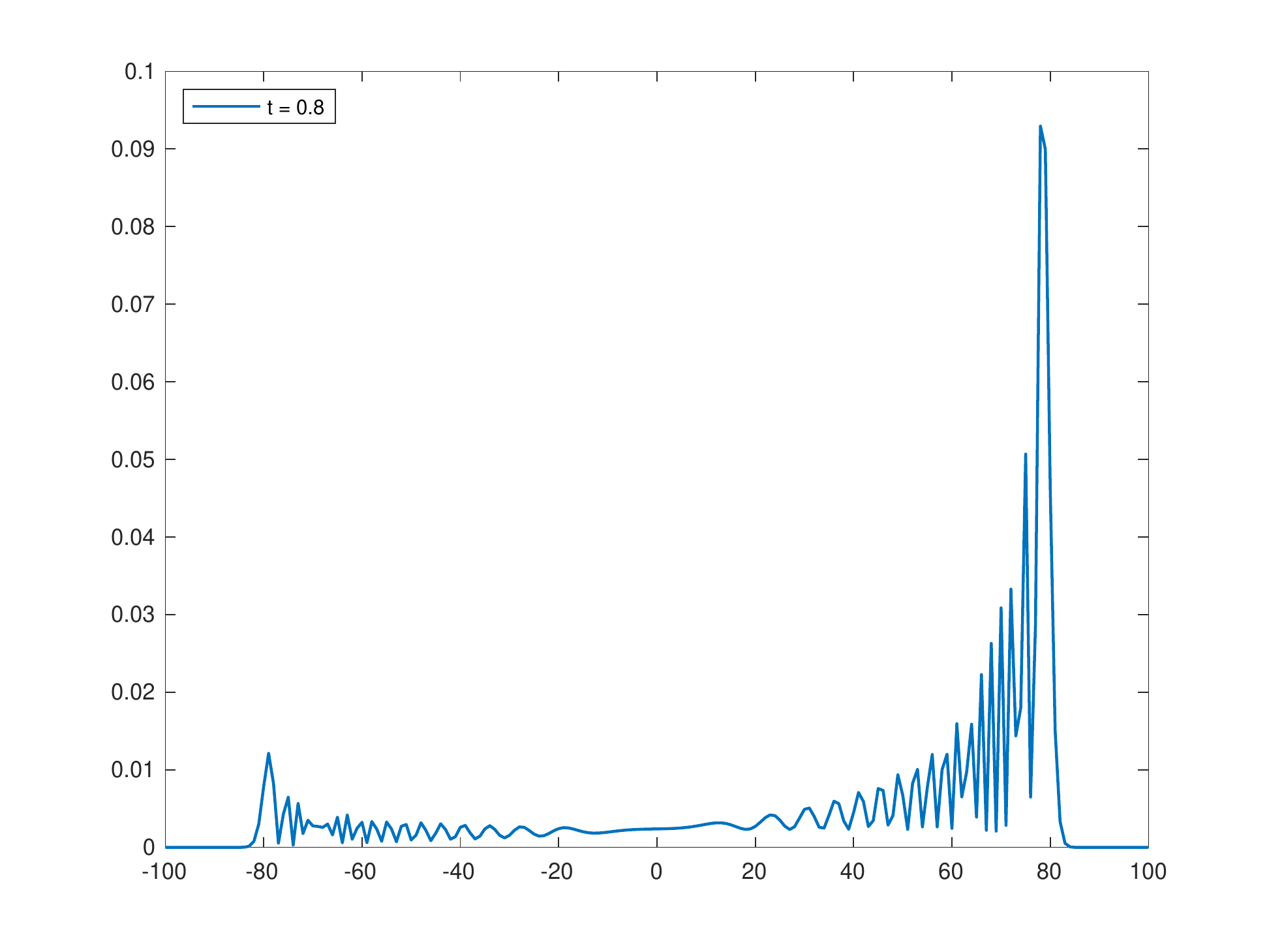}
   \caption{$t=0.8$}
  \label{fig:sub1}
\end{subfigure}%
\begin{subfigure}{.33\textwidth}
  \centering
 \includegraphics[width=1.1\linewidth]{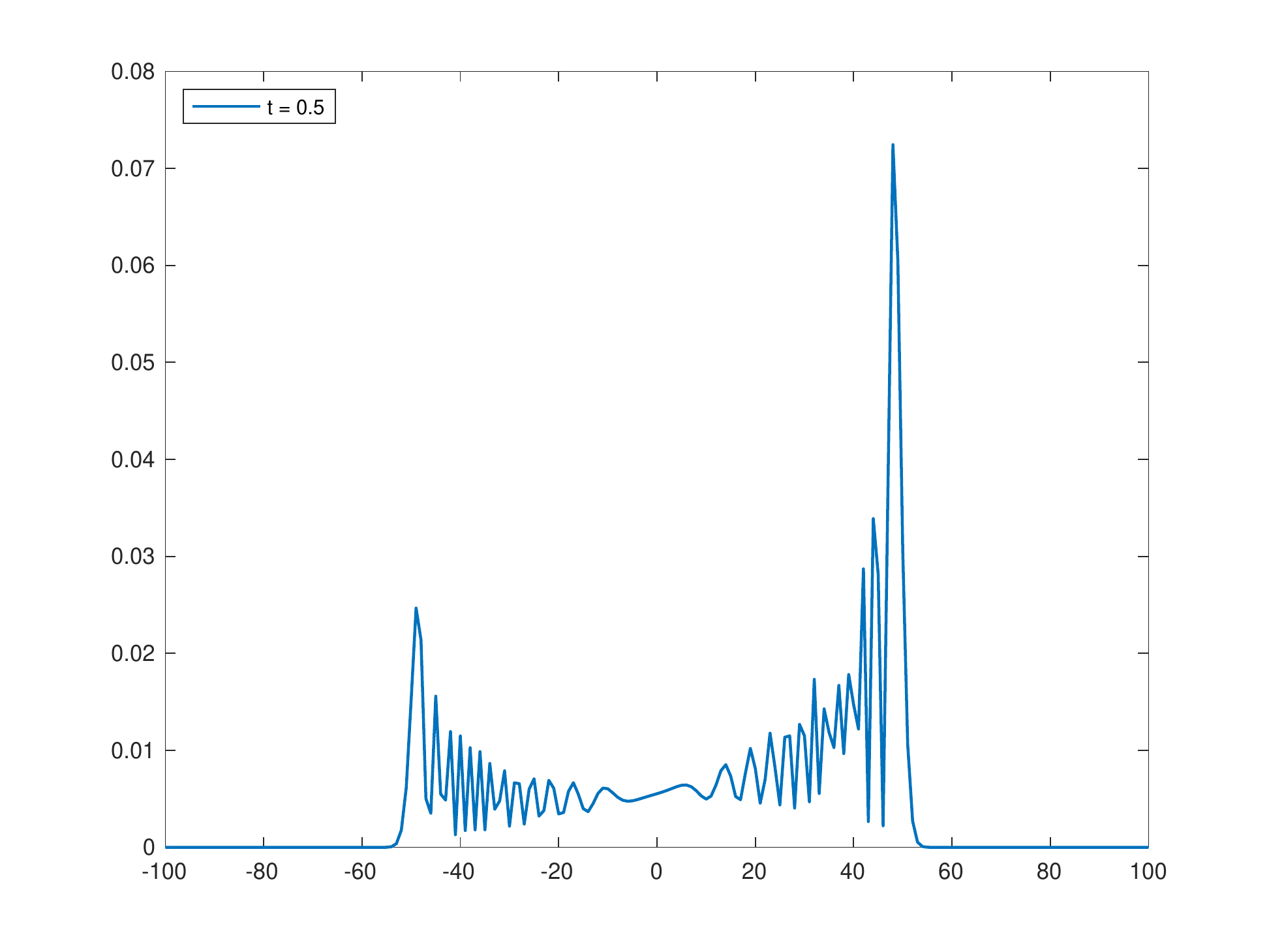}
  \caption{$t=0.5$}
  \label{fig:sub2}
\end{subfigure}
\begin{subfigure}{.33\textwidth}
  \centering
 \includegraphics[width=1.1\linewidth]{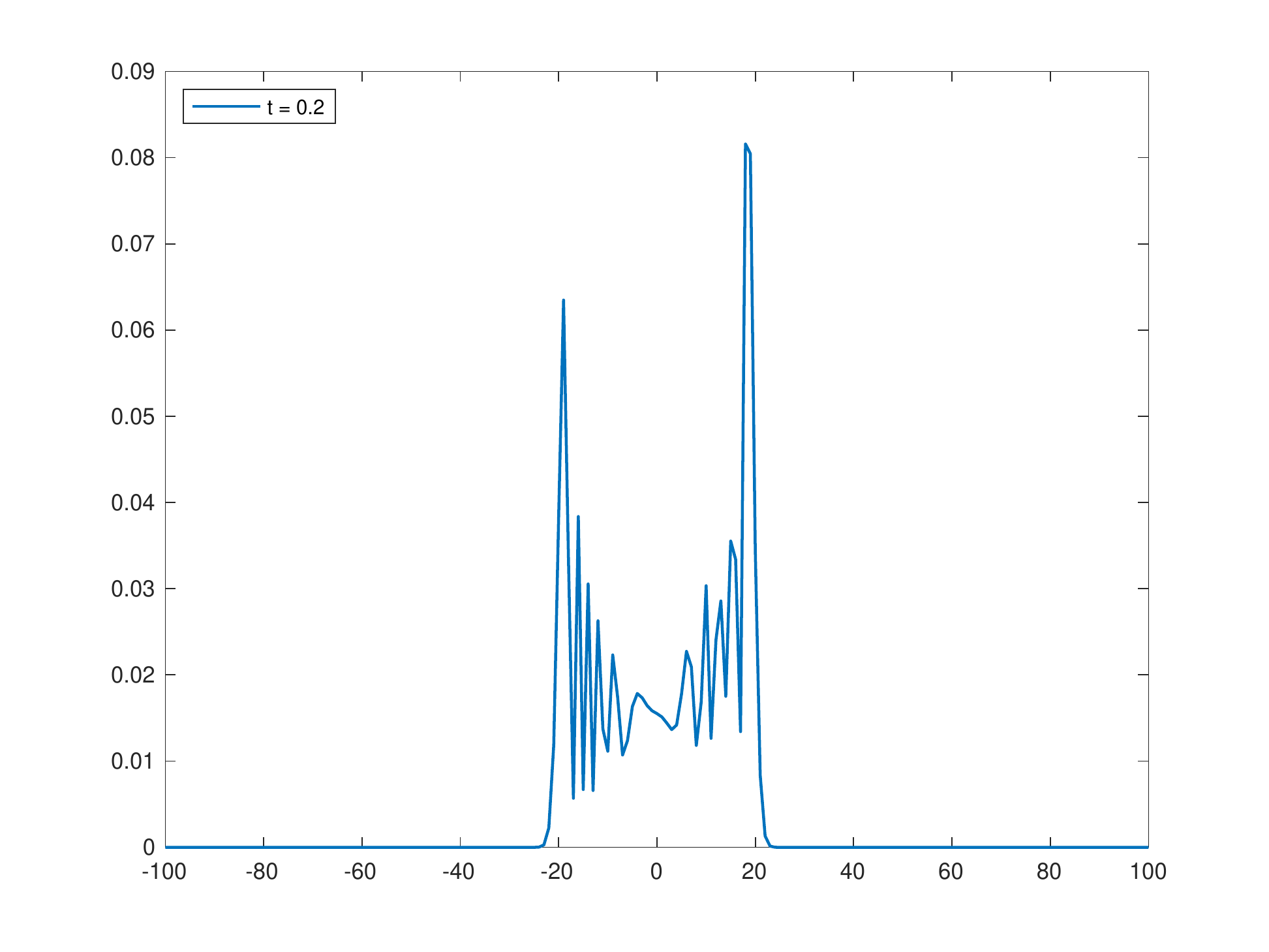}
  \caption{$t=0.2$}
\end{subfigure}
\caption{The probability distribution of the quantum walk generated by $U_t$ in (\ref{U}) after 100 steps starting from $|\uparrow\rangle\otimes|0\rangle$ with different values for the transmission parameter $t$.}
 \label{fig:QW-different-t}
\end{figure}

\section{Main Results}\label{sec:main-results}

Given the quantum walks with parameter $t\in[0,1]$ described in Section \ref{sec:model}, and suppose that the system is initially at the unit vector
\begin{equation}\label{def:Psi}
\Psi_0:=\varphi \otimes |0\rangle \text{ where } \varphi\in \C^2 \text{ and }\|\varphi\|=1.
\end{equation}
We are interested in studying the asymptotic transport of $\Psi_0$ generated by the quantum walk  via 
\begin{equation}
U_{t,\mathcal{D}}=\mathcal{D}V_t W_t
\end{equation}
where $W_t$ and $V_t$ are the coin and the shift operators given in (\ref{def:WV}), (\ref{W-act}), and (\ref{V-act}). While we initially take $\mathcal{D}$ as a general local field (diagonal and unitary), our main result assumes a specific class of periodic local fields, see (\ref{def:Dn}) below.

Towards the asymptotic velocity, we will be particularly interested in the discrete-time evolution of the position operator
\begin{equation}
\tau^{U_{t,\mathcal{D}}}_N(X):=(U_{t,\mathcal{D}}^*)^N X U_{t,\mathcal{D}}^N,\quad N\in\Z,
\end{equation}
where $X$ is the unbounded position operator on $\mathcal{H}
$
\begin{equation}\label{def:X}
X:=\idty_2\otimes \sum_{j\in\Z} |j|\ |j\rangle\langle j| \text{ on } \mathcal{H}.
\end{equation}
$X$ captures to position of the walker regardless of its internal degree. It is easy to see that the domain of $\tau^{U_{t,\mathcal{D}}}_N$ (for any $N$) includes all states  $\phi\otimes\psi\in\mathcal{H}$ where $\phi\in\C^2$ and $\psi$ is finitely supported.

Note that the expected position of $\Psi_0$ after $N$ quantum walk steps determined by $U_{t,\mathcal{D}}$ is
\begin{equation}
\left\langle X\right\rangle_{U_{t,\mathcal{D}}^N  \Psi_0}:=\langle U_{t,\mathcal{D}}^N  \Psi_0,  X U_{t,\mathcal{D}}^N\Psi_0\rangle=\langle   \Psi_0,  \tau^{U_{t,\mathcal{D}}}_N(X)\Psi_0\rangle.
\end{equation}
Hence, the \emph{asymptotic quantum velocity} (or just \emph{velocity}) is defined to be
\begin{equation}\label{def:v}
v_{U_{t,\mathcal{D}}}= \limsup_{N\rightarrow\infty}\frac1N \left\langle X\right\rangle_{U_{t,\mathcal{D}}^N  \Psi_0}.
\end{equation}
Observe that we have an obvious a priori bound for $v_{U_{t,\mathcal{D}}}$ that results from the fact that the support of $U_{t,\mathcal{D}}^N \Psi_0 \in \mathcal{H}$ is  compact. More precisely, (\ref{U-act}) shows that
\begin{equation}\label{def:X-NN}
U_{t,\mathcal{D}}^N \Psi_0=(\idty_2\otimes \idty_{[-N,N]})U_{t,\mathcal{D}}^N \Psi_0,
\end{equation}
where $ \idty_{[-N,N]}$ is the restriction operator to $[-N,N]\cap\Z$ on $\ell^2(\Z)$.
So, by defining $X_{[-N,N]}:=(\idty_2\otimes \idty_{[-N,N]}) X$ and noting that $X_{[-N,N]}/N\leq \idty_\mathcal{H}$, we find that
\begin{equation}\label{eq:v<1}
\frac1N\left\langle X\right\rangle_{U_{t,\mathcal{D}}^N  \Psi_0}=\frac1N\left\langle X_{[-N,N]}\right\rangle_{U_{t,\mathcal{D}}^N  \Psi_0}\leq  \left\langle \idty_\mathcal{H}\right\rangle_{U_{t,\mathcal{D}}^N  \Psi_0}=1,
\end{equation}
meaning  that  $v_{U_{t,\mathcal{D}}}$ cannot be more than one unit per a quantum step for any $t\in[0,1]$ and any local field $\mathcal{D}$. In fact, one can see from (\ref{U:t=1}) that $v_{t=1,\mathcal{D}=\idty}=1$ (this is explained more in Lemma \ref{lem:v<t} below). $t=0$ corresponds to the trivial localization when $U_{0,\mathcal{D}}=\mathcal{D}$ and hence
$v_{U_{0,\mathcal{D}}}=0$ for any local field $\mathcal{D}$.

\subsection{Field-free velocity bounds}

Regardless of the effects of local fields, it is intuitively clear that the velocity depends on the transmission parameter (it ``generates'' the off-diagonal elements of $U_t$). Recall that the simple argument (\ref{eq:v<1}) gives an a priori upper bound of one for the velocity $v_{U_{t,\mathcal{D}}}$. With a little more effort we can show a $t$-dependent bound for $v_{U_{t,\mathcal{D}}}$.
 The following Lemma gives such an upper bound on the velocity of the quantum walk that is independent of the local field $\mathcal{D}$. In particular, Lemma \ref{lem:v<t} shows that the quantum velocity decays at least linearly in the transmission parameter $t$ in the presence of any local field. 

\begin{lem}\label{lem:v<t}
For the quantum walk generated by the unitary $U_{t,\mathcal{D}}=\mathcal{D}V_t W_t$ we have
\begin{equation}\label{eq:lem:v<t}
\frac1N \left\langle X\right\rangle_{U_{t,\mathcal{D}}^N  \Psi_0}\leq \frac1N\|\tau^{U_{t,\mathcal{D}}}_N(X)-X\|\leq t
\end{equation}
for any $t\in[0,1]$, $N\in\Z^+$, and any local field $\mathcal{D}$. Moreover, the two inequalities in  (\ref{eq:lem:v<t}) become equalities in the extreme cases $t=0$ and $t=1$.
\end{lem}
The proof is simple and it is presented in Appendix  \ref{App:proof:v<t}. We note that the first inequality in (\ref{eq:lem:v<t}) is a straightforward consequence of the fact that $X\Psi_0=0$,
\begin{equation}
\left\langle X\right\rangle_{U_{t,\mathcal{D}}^N\Psi_0} = \left\langle\tau^{U_{t,\mathcal{D}}}_N(X)-X\right\rangle_{\Psi_0}
\leq \|\tau^{U_{t,\mathcal{D}}}_N(X)-X\|.
\end{equation}
It is worth mentioning here that $\|\tau^{U_{t,\mathcal{D}}}_N(X)-X\|$ is our central object of interest in bounding the velocity in this work.

Lemma \ref{lem:v<t} is showing that the velocity of the quantum walk described by $U_{t,\mathcal{D}}$ cannot be more than $t$ in the presence of any local field. i.e.,
\begin{equation}\label{eq:v<t}
\sup_{\mathcal{D}} v_{U_{t,\mathcal{D}}}\leq t, \
\end{equation}
where the supremum is taken over all local fields $\mathcal{D}$ of the form given in (\ref{def:D}). We remark here that the linear bound in $t$ in (\ref{eq:lem:v<t}) applies trivially to higher moments: For any $p\in\Z^+$,
\begin{equation}\label{eq:highmoments}
\frac1{N^p} \left\langle X^p\right\rangle_{U_{t,\mathcal{D}}^N  \Psi_0}=\frac1{N^p} \left\langle X_{[-N,N]}^p\right\rangle_{U_{t,\mathcal{D}}^N  \Psi_0}\leq \frac1N\left\langle X\right\rangle_{U_{t,\mathcal{D}}^N\Psi_0}\leq t
\end{equation}
where $X_{[-N,N]}$ is defined in (\ref{def:X-NN}) and we used the fact that $X^{p-1}_{[-N,N]}/N^{p-1}\leq \idty_{\mathcal{H}}$.

Lemma \ref{lem:v<t} applies to any local field $\mathcal{D}$. Our main result, Theorem \ref{thm:main}, shows that when the local field is chosen to be of a certain periodic form, the quantum velocity can be made arbitrarily small for any fixed $t$ in a certain range. In the following we define the corresponding local field.

\subsection{Velocity bounds in the presence of a periodic local field}

For a fixed $n\in\Z^+$, we require that the local field (diagonal and unitary) $\mathcal{D}_n$ has the following properties:
\begin{itemize}
\item[(a)] Periodic of (fundamental) period $n$.
\item[(b)] $\exists c\in\C\setminus\{1\}$; $\forall j\in\Z$ we have $\displaystyle \langle\delta_{j+1},\mathcal{D}_n\delta_{j+1}\rangle=c\langle \delta_{j},\mathcal{D}_n \delta_{j}\rangle$.
\end{itemize}
In fact, properties (a) and (b) above mean that the most general form of $\mathcal{D}_n$ is given as
\begin{equation}\label{def:Dn}
\mathcal{D}_n=\beta \left(\diag\{1,\alpha,\alpha^2,\ldots,\alpha^{n-1}\}\right)^{\oplus\Z},
\end{equation}
where $\beta,\ \alpha\in\C$ with $|\beta|=|\alpha|=1$ and $\alpha\neq 1$ is an $n$-th root of the unity (not necessarily the standard primitive $n$-th root of unity). Hence $1, \alpha,\alpha^2, \ldots, \alpha^{n-1}$ are distinct, that is,
\begin{equation}\label{def:alpha}
\alpha=e^{2\pi k i/n}\text{ for some } k\in\Z^+ \text{ such that } \gcd(k, n)=1.
\end{equation}
 $\beta$ in (\ref{def:Dn}) is just a trivial phase, and without loss of generality we assume that $\beta=1$.  The direct sum in (\ref{def:Dn}) corresponds to a block diagonal matrix representation of $\mathcal{D}_n$ where each block is the diagonal matrix $\diag\{1,\alpha,\ldots,\alpha^{n-1}\}$ 
 with respect to the basis $\{\delta_{j+k},\, k=0,1,2,\ldots,n-1\}$ for any fixed starting index $j\in \Z$. Different such $j$'s produce just trivial phases. Thus we will assume that $\mathcal{D}_n$ is given as
 \begin{equation}
 \mathcal{D}_n\delta_{j}= \alpha^j \delta_{j}, \ j\in\Z.
 \end{equation}

Note that
$
\mathcal{D}_n^*=\overline{\mathcal{D}_n}=\mathcal{D}_n^{-1}$, $\mathcal{D}_n^n=\idty_{\mathcal{H}}$, and $\mathcal{D}_n^j\neq \mathcal{D}_n^k$ for all $j\neq k$ in $\{1,\ldots,n\}$. Moreover, observe that $\mathcal{D}_n$ acts nontrivially on both components of the tensor product: The spin and the position of the walker, i.e.,
\begin{equation}\label{def:Dn-2}
\mathcal{D}_n=\begin{pmatrix}1 & 0\\0& \alpha\end{pmatrix}\otimes \sum_{j\in\Z} \alpha^{2j}|j\rangle\langle j|=
\begin{pmatrix}1 & 0\\0& \alpha\end{pmatrix}\otimes 
\left(\diag\{1,\alpha^2,\ldots, \alpha^{2(n-1)}\}\right)^{\oplus\Z}.
\end{equation}
In particular, for any $j\in\Z$
\begin{eqnarray}
\mathcal{D}_n&:& |\uparrow\rangle\otimes |j\rangle\,  \mapsto\,  \alpha^{2j}|\uparrow\rangle\otimes |j\rangle\nonumber\\
&& |\downarrow\rangle\otimes |j\rangle\,  \mapsto\,  \alpha^{2j+1}|\downarrow\rangle\otimes |j\rangle.
\end{eqnarray}
Our main result considers the quantum walk generated by unitary 
\begin{equation}\label{def:U-Dn}
U_{t,n}:=U_{t,\mathcal{D}_n}=\mathcal{D}_n V_t W_t
\end{equation}
 with the $n$-periodic local field $\mathcal{D}_n$ defined in (\ref{def:Dn}), the coin unitary $W_t$, and the shift operator $V_t$ given in (\ref{def:WV}) or by their actions on the basis of $\mathcal{H}$ in (\ref{W-act}) and (\ref{V-act}). 
 \begin{rem}
 Note that the application of $\mathcal{D}_n$ to the product of the coin and shift operators $W_t$ and $V_t$ can be regarded as a quantum walk with an identical shift operator at every position in $\Z$ and a coin operator that depends on the position of the walker. More precisely, it is direct to check that
\begin{equation}
U_{t,n}=\mathcal{D}_n V_t W_t=V_t^{(\alpha)} W_t^{(\alpha)}
\end{equation}
where 
\begin{eqnarray}
V_t^{(\alpha)}&=&r\idty_{\mathcal{H}}+\alpha^{-1} t|\downarrow\rangle\langle\uparrow|\otimes T_{-1}
-\alpha t|\uparrow\rangle\langle\downarrow|\otimes T_1 \nonumber\\
W_t^{(\alpha)}&=&\begin{pmatrix} r & t\\ -\alpha t & \alpha r\end{pmatrix}\otimes \sum_{j\in\Z}\alpha^{2j}|j\rangle\langle j|.
\end{eqnarray}
That means that $W_t^{(\alpha)}$ is a coin operator that assigns a phase to the walker depending on the its position in $\Z$, for example
\begin{equation}
W_t^{(\alpha)}: |\uparrow\rangle\otimes |j\rangle\, \mapsto\, \alpha^{2j}\left(r|\uparrow\rangle-\alpha t|\downarrow\rangle\right)\otimes |j\rangle.
\end{equation}
We think that the dependency of $W_t^{(\alpha)}$ on the position allows for the reduction of velocity result in Theorem \ref{thm:main}. It is noteworthy here to mention that some non identical random coin operators are considered in \cite{JM10, Kono09}.  In \cite{JM10}, the coin operators are chosen to be a family of i.i.d. random operators with some general requirements, and the randomness is determined by phases carried by all matrix elements of the coin operator. Those non-identical random coin operators lead to dynamical localization everywhere in the spectrum, for almost all realizations of coin operators. By contrast, \cite{Kono09} shows that a certain choice of random coin operators does not lead to localization. The random coin operators in \cite{Kono09} are given by a unitary matrix with random phases only on the diagonal. The lack of localization in this disordered case is explained by the existence of gauge transformations that fully eliminates the randomness in the model.

Note here that $W_t^{(\alpha)}$ is periodic on $\Z$ in the since that 
\begin{equation}\label{eq:Wt-period}
 \left\langle s,j+n\left |W_t^{(\alpha)}\right| s, j+n \right\rangle=  \left\langle s,j\left |W_t^{(\alpha)}\right| s, j \right\rangle, \text{ for all }j\in\Z, \text{ and } s\in\{\uparrow, \downarrow\}.
\end{equation}
Here $|s,j\rangle :=|s\rangle\otimes |j\rangle$. Only when the periodicity length $n$ goes to infinity we guarantee by Theorem \ref{thm:main} below  a zero asymptotic velocity. It seems that this periodicity  (\ref{eq:Wt-period}) of  the coin operator on $\Z$ prevents the quantum walk to have zero asymptotic velocity for any finite period $n$, see  Conjecture \ref{conj} below.
 \end{rem}

We denote the discrete time evolution of the quantum walk with periodic local field $\mathcal{D}_n$ by 
\begin{equation}
\tau_N^{U_{t,n}}(X):=(U_{t,n}^*)^N X U_{t,n}^N, \quad N\in \N
\end{equation}
and the associated quantum walk velocity (\ref{def:v})  by $v_{U_{t,n}}$,
\begin{equation}
v_{U_{t,n}}:=v_{U_{t,\mathcal{D}_n}}= \limsup_{N\rightarrow\infty}\frac1N \left\langle X\right\rangle_{U_{t,\mathcal{D}_n}^N  \Psi_0},
\end{equation}
where $\Psi_0$ is defined in (\ref{def:Psi}).
The following theorem shows that the periodic local field $\mathcal{D}_n$ slows down the quantum velocity \emph{exponentially} in the period $n$ when the transmission parameter $t$ is sufficiently small.
\begin{thm}\label{thm:main}
For any fixed $n\in\N$, and the corresponding local $n$-periodic field $\mathcal{D}_n$ defined in (\ref{def:Dn}), the velocity of the quantum walk generated by $U_{t,n}=\mathcal{D}_n V_t W_t$, given in (\ref{def:WV}), satisfies the bound
\begin{equation}\label{eq:main-result}
v_{U_{t,n}}\leq\limsup_{N\rightarrow\infty}\frac{1}{N}\|\tau_N^{U_{t,n}}(X)-X\|\leq (4t)^n,
\end{equation}
where  $t\in[0,1]$ is the transmission parameter.

That is, for any $v_0>0$ there exists an $n$-periodic field $\mathcal{D}_n$ such that for any $t\in(0,1/4)$ the quantum velocity $v_{U_{t,n}}$ satisfies the bound $v_{U_{t,n}}\leq v_0$.
\end{thm}
The case $n=1$ corresponds to $\mathcal{D}_n=\idty_{\mathcal{H}}$ and Theorem \ref{thm:main} produces a linear bound in line with Lemma \ref{lem:v<t}. Our bound in (\ref{eq:main-result}) is irrelevant when $t\geq 1/4$, and we think that the validity of Theorem \ref{thm:main} extends to all $t\in(0,1)$, see Conjecture \ref{conj} below.

Lemma \ref{lem:v<t} guarantees that the quantum velocity can be made arbitrarily small by choosing the corresponding small transmission parameter $t$. It also says that the velocity bound decays at least linearly in $t$. On the other hand, Theorem \ref{thm:main} shows that for a fixed transmission parameter $t\in(0,1/4)$, a specific $n$-periodic field can be activated to guarantee an exponential decay of the velocity,  in the period $n$,  of the quantum walk.

The proof of Theorem \ref{thm:main} follows from the following two main steps:
\begin{itemize}

\item[]\emph{Step 1: (An upper bound for a subsequence of average velocities)} We prove that for every fixed $n\in\N$, $t\in(0,1)$, and the corresponding local $n$-periodic field
 $\mathcal{D}_n$ we have
\begin{equation}\label{main:pf:step1}
\frac{1}{nk}\|\tau_{nk}^{U_{t,n}}(X)-X\|\leq \frac34(4t)^n
\end{equation}
for every $k\in\Z^+$. This step is the main challenge and its proof spans all following sections.

\item[]\emph{Step 2: (Interpolation with the linear bound)} We interpolate  between (\ref{main:pf:step1}) and the linear bound (\ref{eq:lem:v<t}) in Lemma \ref{lem:v<t} to show that
\begin{equation}\label{main:pf:step2}
\frac1N\|\tau_{N}^{U_{t,n}}(X)-X\|\leq (4t)^n\ \text{ for all }\ N\geq N_0
\end{equation}
for some $N_0=N_0(n,t)\in\Z^+$.
\end{itemize}
We present the proof of \emph{Step 2} later in this section, and we comment now on its main result (\ref{main:pf:step2}). Note that (\ref{main:pf:step2}) implies that for any given $t\in[0,1]$ and a periodic field $\mathcal{D}_n$ (\ref{def:Dn}) with period $n$, there exists $N_0(t,n)\in\Z^+$ such that 
\begin{equation}
\frac1N\langle X\rangle_{U_{t,n}^N\Psi_0}\leq (4t)^n, \text{ for all }N\geq N_0.
\end{equation}
That is, when we consider higher moments, the argument in (\ref{eq:highmoments}) gives directly the same bound, i.e., for any $p\in\Z^+$
\begin{equation}
\limsup_{N\rightarrow \infty}\frac1{N^p}\langle X^p\rangle_{U_{t,n}^N\Psi_0}\leq (4t)^n.
\end{equation}
Figure \ref{fig:QW-different-n} reflects the result of Theorem \ref{thm:main} and it demonstrates  how the application of a periodic field (\ref{def:Dn}) of different periods  slows down the transport of the quantum walk when $t=0.2<1/4$ exponentially in the period $n$. 

\begin{figure}[h]
\centering
\begin{subfigure}{.33\textwidth}
  \centering
  \includegraphics[width=1.1\linewidth]{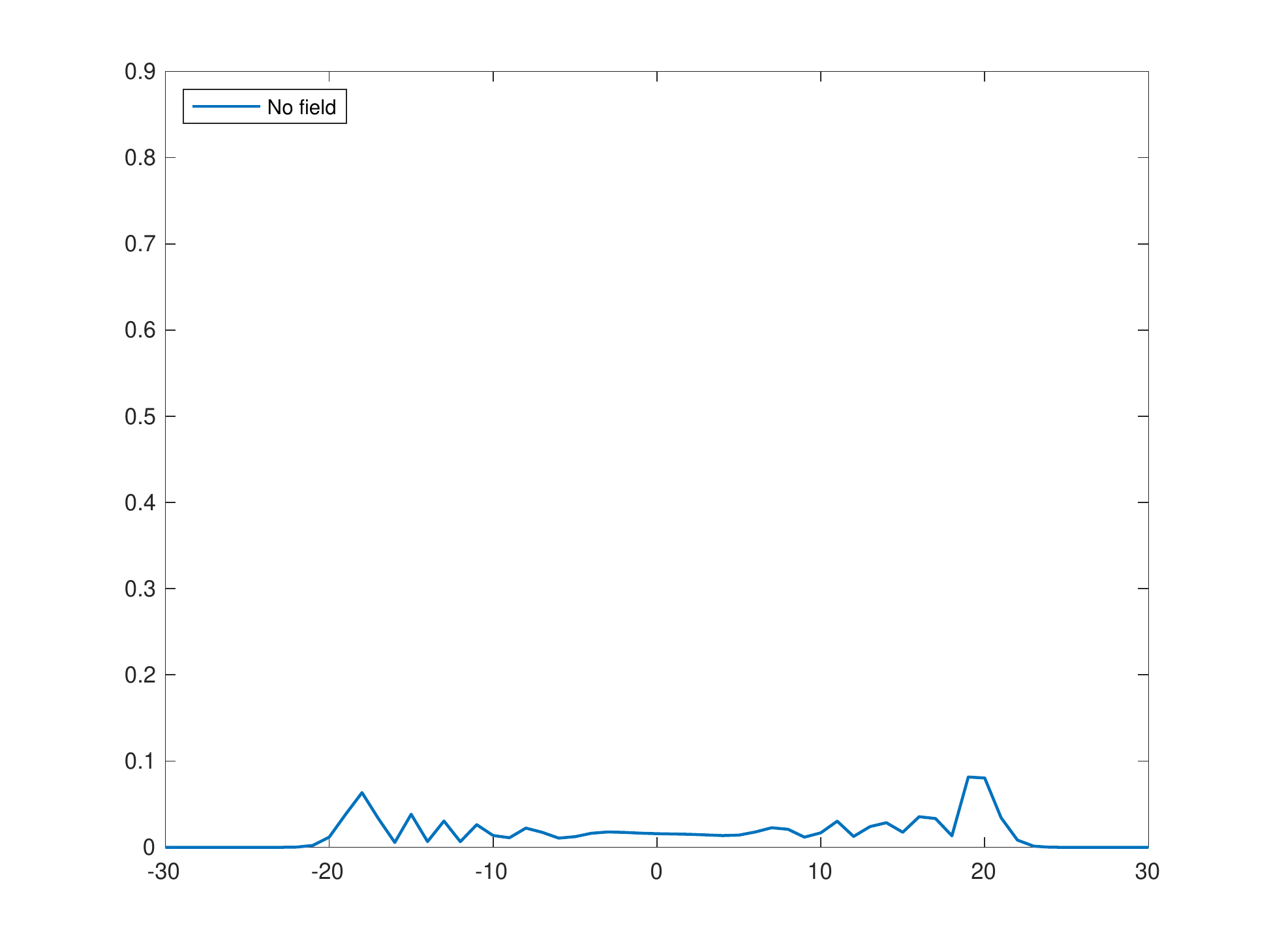}
  \caption{No field, $t=0.2$}
  \label{fig:sub1}
\end{subfigure}%
\begin{subfigure}{.33\textwidth}
  \centering
  \includegraphics[width=1.1\linewidth]{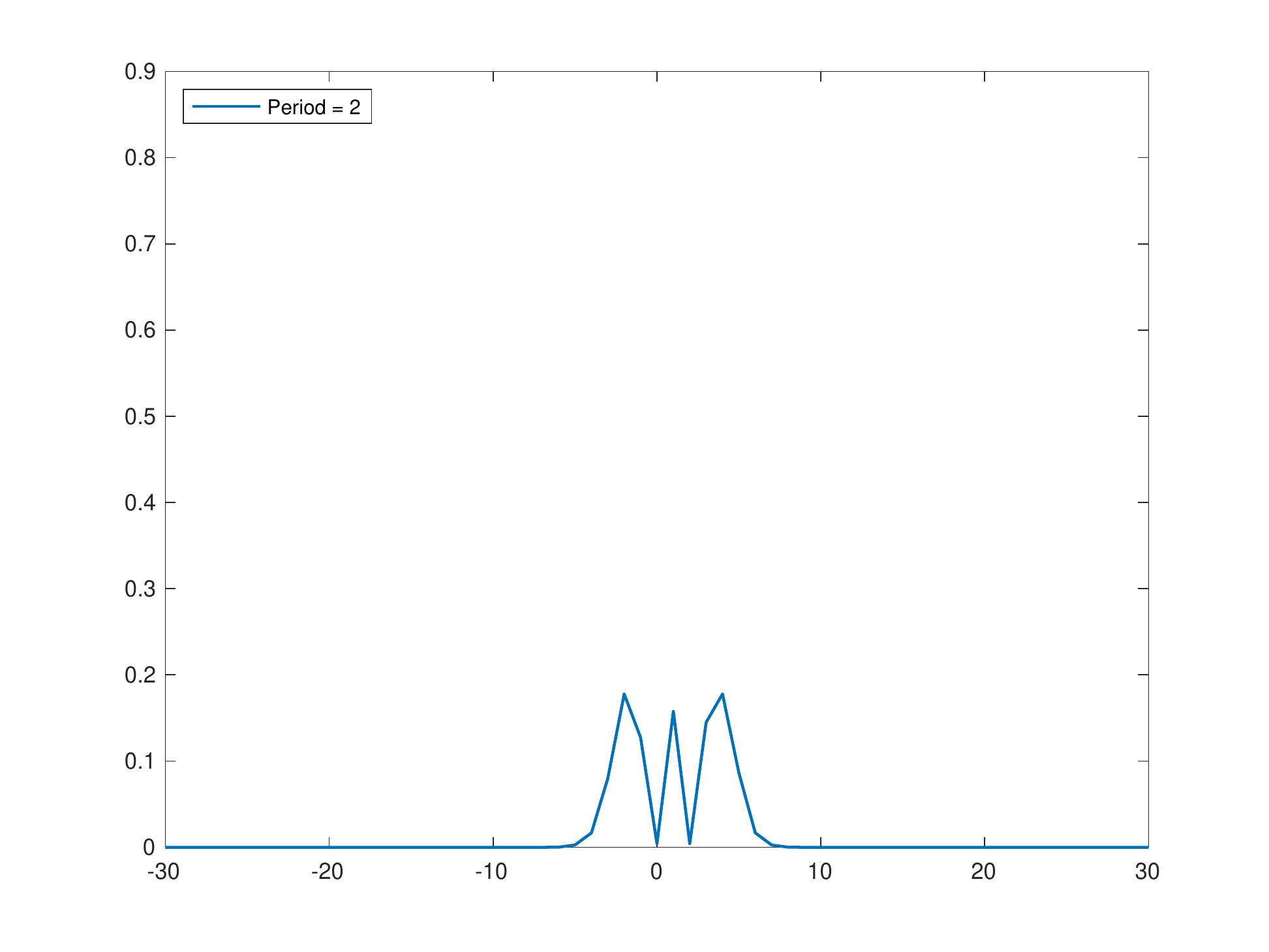}
  \caption{$n=2$, $t=0.2$}
  \label{fig:sub2}
\end{subfigure}
\begin{subfigure}{.33\textwidth}
  \centering
   \includegraphics[width=1.1\linewidth]{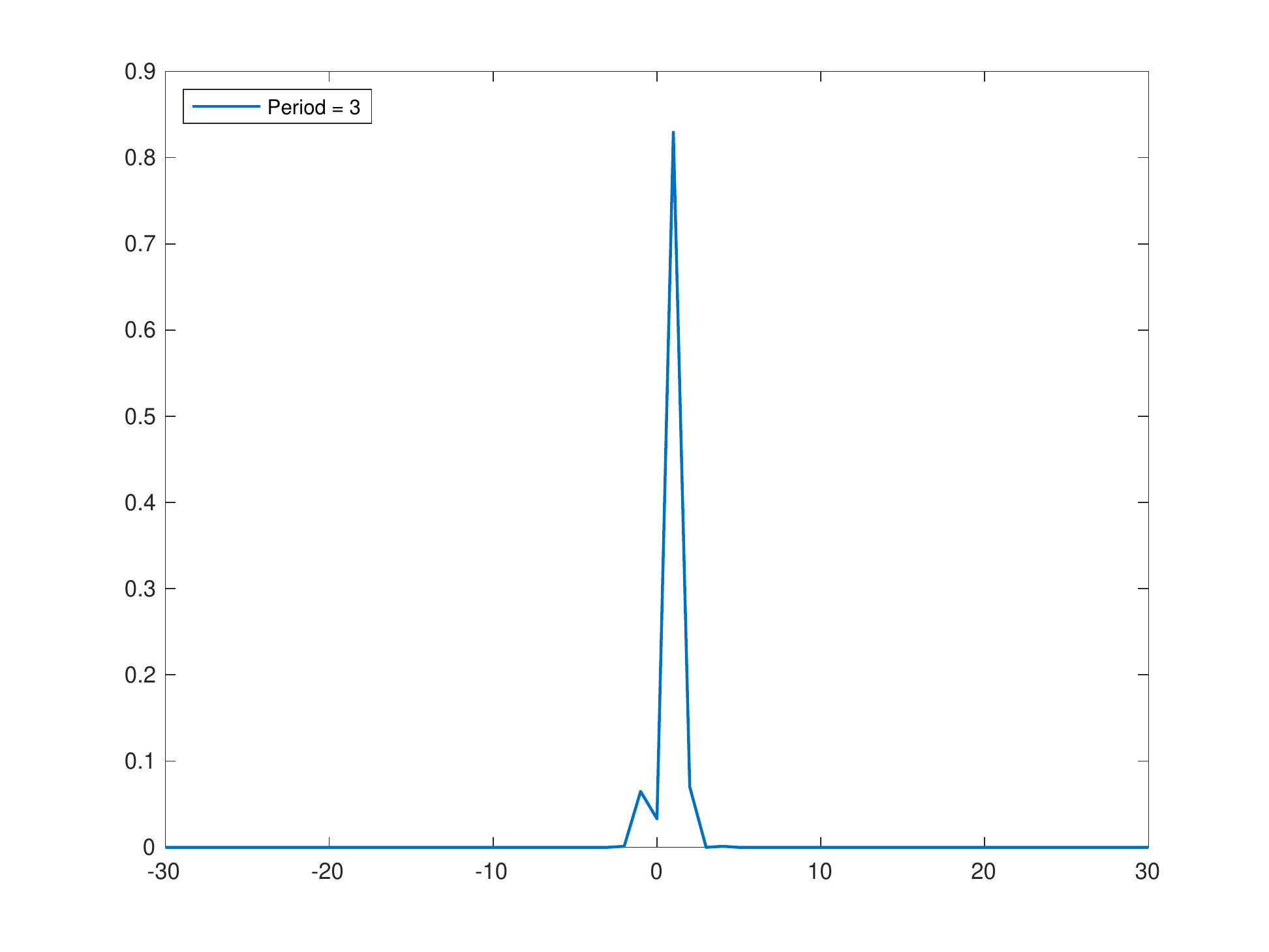}
   \caption{$n=3$, $t=0.2$}
\end{subfigure}
\caption{The probability distribution of the quantum walk generated by $U_{t=0.2,n}$  after 100 steps starting from $|\uparrow\rangle\otimes|0\rangle$ with different periods $n$. The graphs are restricted to the position window $[-30,30]$.}
\label{fig:QW-different-n}
\end{figure}

A natural question to ask is whether the condition that $t\in(0,1/4)$ in Theorem \ref{thm:main} can be relaxed or removed. Our numerics are showing that periodic fields of the form (\ref{def:Dn}) indeed slow down the quantum velocities when $t\in[1/4,1)$, see for example Figure \ref{fig:QWdifferent-n-t=0.8} that shows three probability distributions for the quantum walk after 100 steps starting from  $|\uparrow\rangle\otimes|0\rangle$  when the transmission parameter is fixed to be $t=0.8$ and the local field is periodic with an increasing periods $n=1$, $n=5$, and $n=10$. It is interesting to see that the position with the highest probability in the three cases  are about 80, 30 and 10, respectively, in perfect match with the $100(0.8)^n$ for $n=1,5,10$. This leads up to the following conjecture.
\begin{conj}\label{conj}
For any $t\in(0,1)$, $n\in\N$, and the corresponding local $n$-periodic field $\mathcal{D}_n$ defined in (\ref{def:Dn}), the velocity of the quantum walk $U_{t,n}$ scales like $t^n$, i.e.,
 \begin{equation}
v_{U_{t,n}}= \mathcal{O}(t^n).
\end{equation}
\end{conj}

\begin{figure}[h]
\centering
\begin{subfigure}{.33\textwidth}
  \centering
\includegraphics[width=1.1\linewidth]{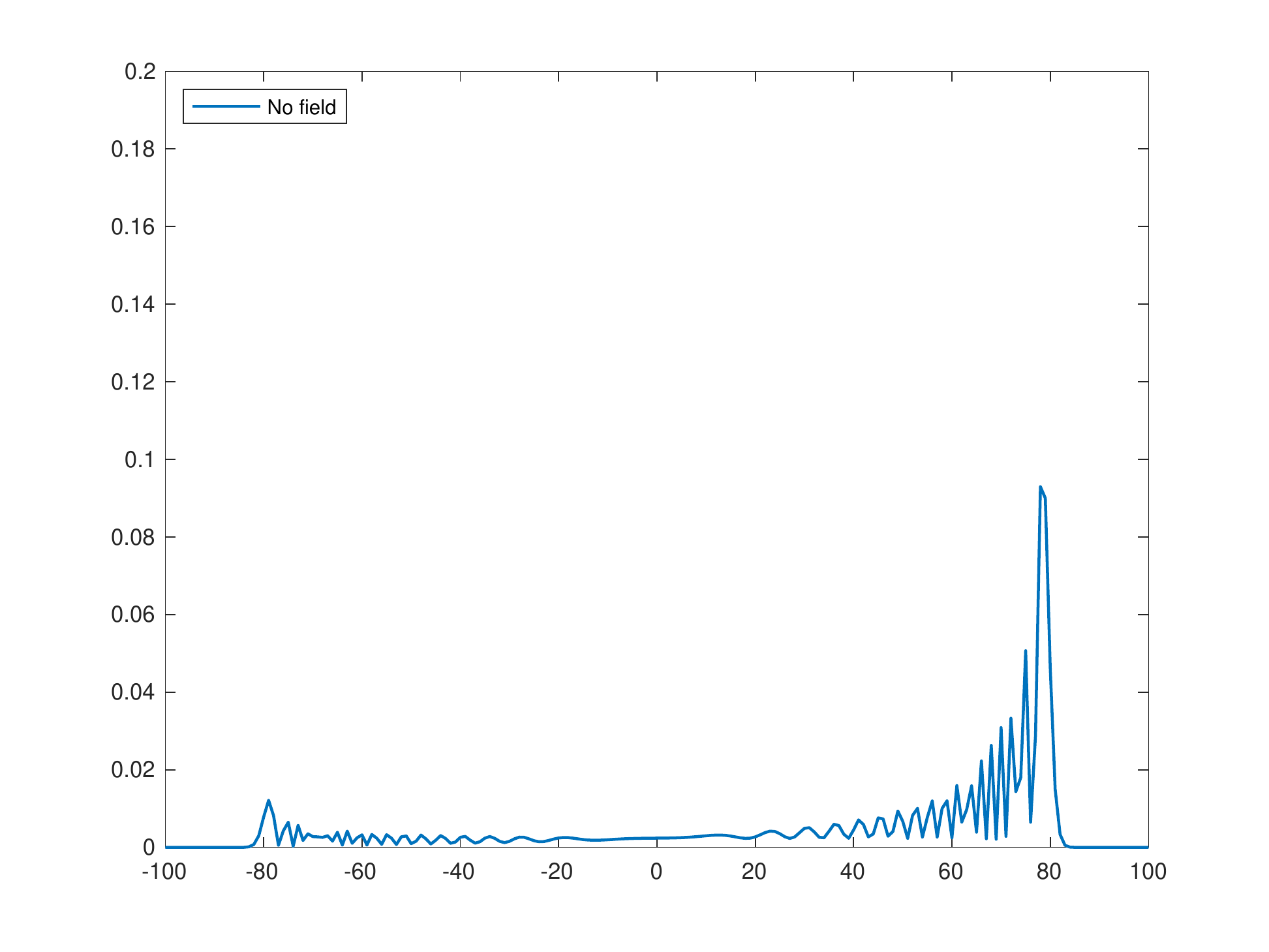}
  \caption{No field, $t=0.8$}
  \label{fig:sub1}
\end{subfigure}%
\begin{subfigure}{.33\textwidth}
  \centering
  \includegraphics[width=1.1\linewidth]{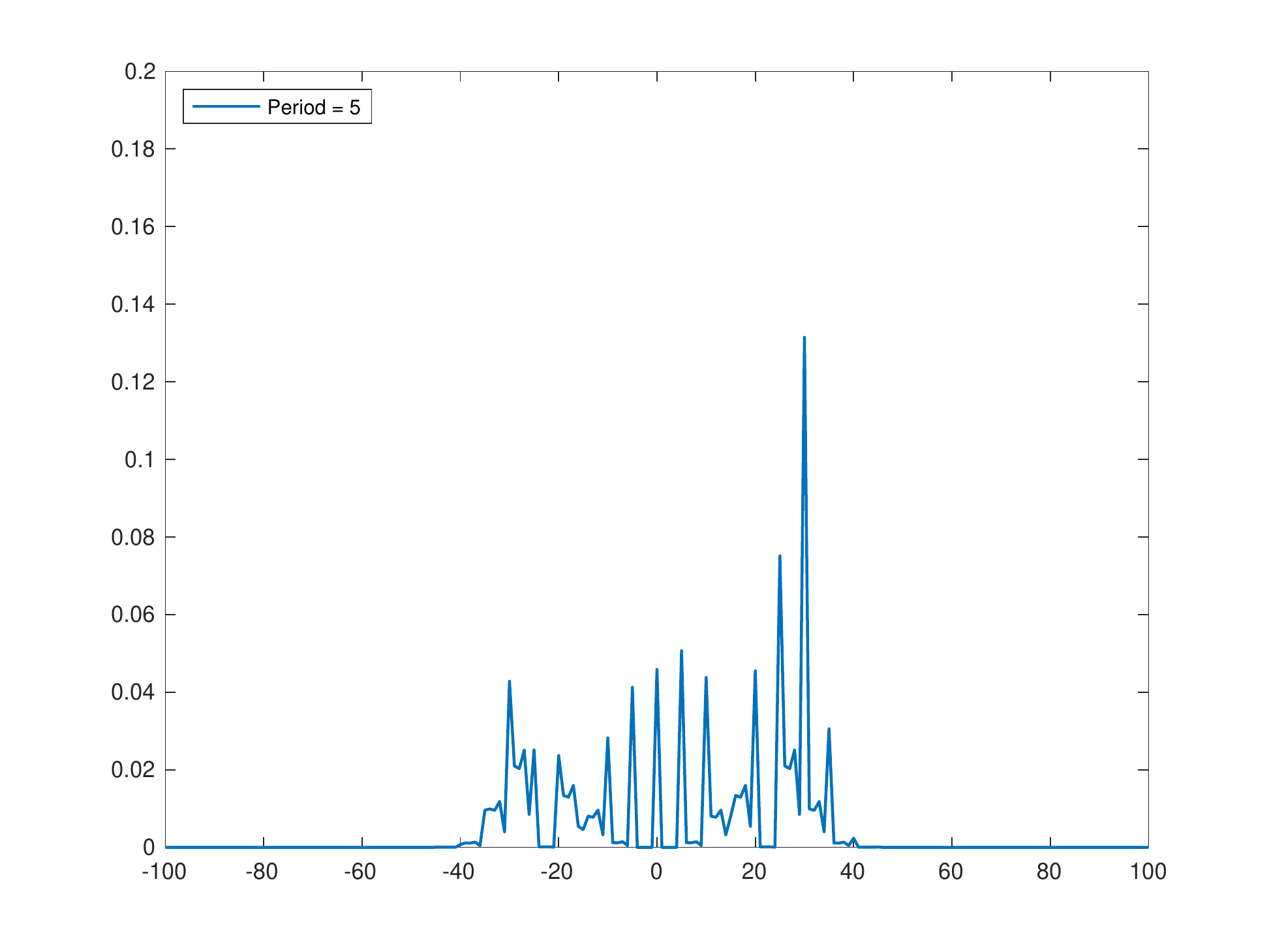}
  \caption{$n=5$, $t=0.8$}
  \label{fig:sub2}
\end{subfigure}
\begin{subfigure}{.33\textwidth}
  \centering
   \includegraphics[width=1.1\linewidth]{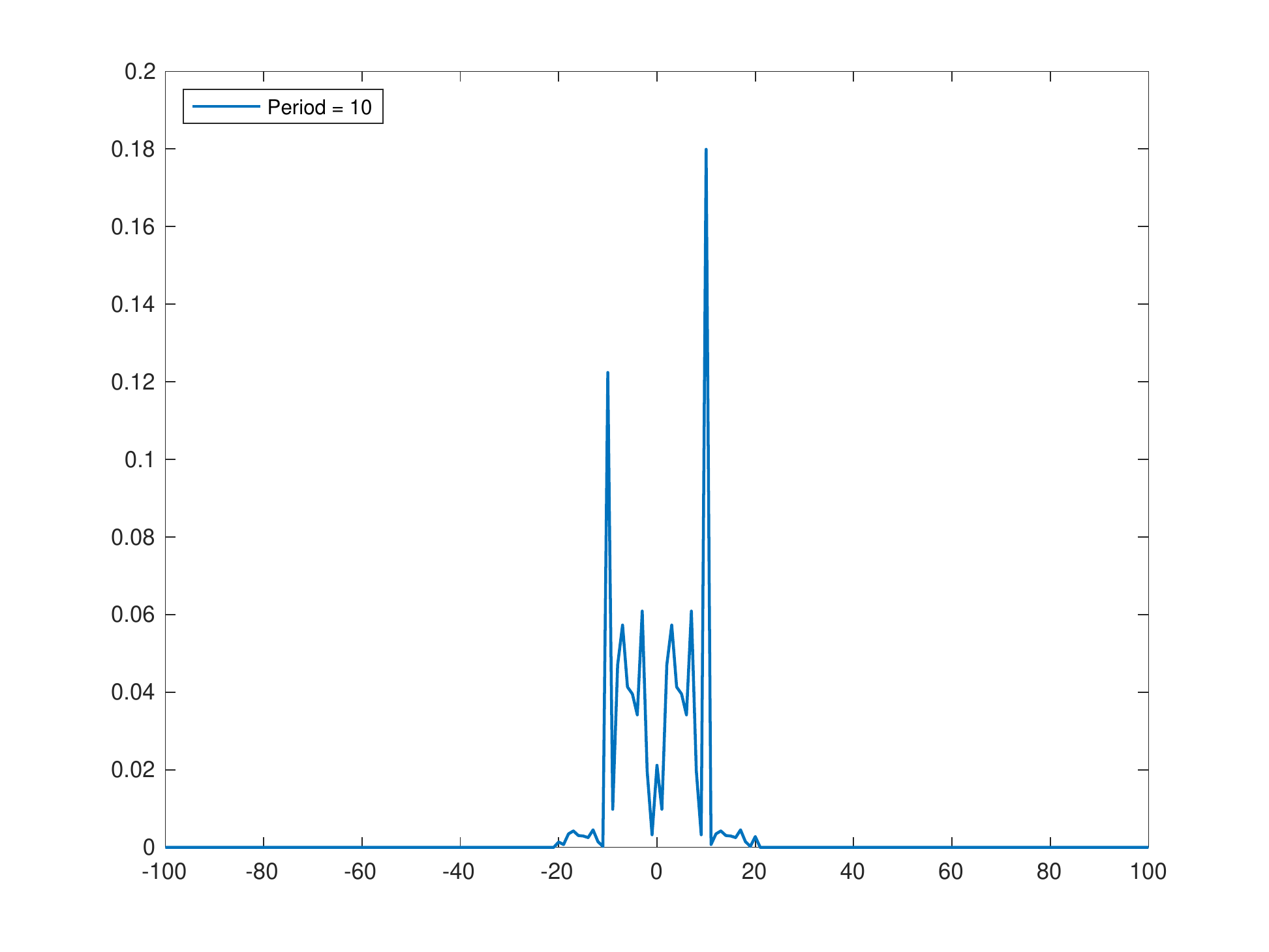}
   \caption{$n=10$, $t=0.8$}
\end{subfigure}
\caption{The probability distribution of the quantum walk generated by $U_{t=0.8,n}$  after 100 steps starting from $|\uparrow\rangle\otimes|0\rangle$ with different periods $n$.}
    \label{fig:QWdifferent-n-t=0.8}
\end{figure}

We present now the interpolation argument (\emph{Step 2}) while assuming the result of \emph{Step 1} (\ref{main:pf:step1}) whose proof spans the rest of this work.

In the following we prove that for any fixed $t\in(0,1)$ and any integer $n>1$, there exists an $n$-periodic field $\mathcal{D}_n$  such that
\begin{equation}
\frac{1}{N}\|\tau^{U_{t,n}}_N(X)-X\|\leq (4t)^n \text{\quad for all }N> N_0(t,n):=\max\left\{\frac{n}{(4t)^{n-1}},n\right\}.
\end{equation}
Note that 
for any integer $N> N_0(t, n)$, there exists a unique $\hat k\in\N$ and $r\in\{0,1,\ldots, n-1\}$ such that
\begin{equation}\label{eq:N-dec}
N=n\hat k+r.
\end{equation}
Add an subtract $\tau^{U_{t,n}}_r(X)$  to $\tau^{U_{t,n}}_N(X)-X$ to find
\begin{eqnarray}\label{eq:interp:Tri}
\|\tau^{U_{t,n}}_N(X)-X\|&\leq&
\|\tau^{U_{t,n}}_{n\hat k}(X)-X\|+\|\tau^{U_{t,n}}_r(X)-X\|.
\end{eqnarray}
We use the bound  $(\ref{main:pf:step1})$ in the first term, and we use the linear bound (\ref{eq:v<t})
\begin{equation}
\|\tau^{U_{t,n}}_r(X)-X\|=\|[X,U_{t,n}^r]\|\leq r t
\end{equation}
in the second term of (\ref{eq:interp:Tri}) to obtain
\begin{equation}\label{interp:1}
\frac{1}{N}\|\tau^{U_{t,n}}_N(X)-X\|\leq \frac{3}{4N} n\hat k\ (4t)^n+\frac{r}{N}t.
\end{equation}
Next, we consider the two cases for $N_0(t, n)$. If $N_0(t,n)=n$, i.e., if $t\in[1/4,1)$, then we simply use $1\leq t\leq (4t)^{n}$, and (\ref{interp:1}) reads as
\begin{equation}
\frac{1}{N}\|\tau^{U_{t,n}}_N(X)-X\|\leq \frac{1}{N} \left(\frac34 n\hat k+r\right)\ (4t)^n\leq (4t)^n.
\end{equation}
For the other case when $N_0=n(4t)^{1-n}$, i.e., $t\in(0,1/4)$, we use
$N> n(4t)^{1-n}$ in the second term in (\ref{interp:1}), we find
\begin{eqnarray}
\frac{1}{N}\|\tau^{U_{t,n}}_N(X)-X\|&\leq& \frac{1}{N} \left(\frac34 n\hat k+\frac{r N}{4n}\right)(4 t )^n
\leq \frac{1}{N} \left(\frac32n\hat k+\frac{N}{4}\right)(4 t )^n \nonumber\\
&\leq&(4t)^n
\end{eqnarray}
where we used the fact that $r< n$, see (\ref{eq:N-dec}).

\section{An upper bound for a subsequence of average velocities}\label{sec:subsequence}

In this section we present the main steps and tools to show the bound in (\ref{main:pf:step1}) for a subsequence of average velocities. First, we explain how the average velocity on the LHS of (\ref{main:pf:step1}) is bounded by $\frac{1}{n}\|[X,U_{t,n}^n]\|$, where we recall here that $U_{t,n}=\mathcal{D}_n U_t$. In Section \ref{subsec:BC}, we write $U_t$ as a sum of three operators, each corresponds to a certain power $0,1,2$ of the transmission parameter $t$. Then in Sections \ref{subsec:s} and \ref{subsec:expand} we expand $(\mathcal{D}_nU_t)^n$ in terms of the so called  \emph{elementary symmetric polynomials} of (non-commuting) operators $\s_{[0,n-1]}^{\ell, m}$ that will represent the main object to understand. The crucial fact, Theorem \ref{thm:main-step}, says that those symmetric polynomials of operators that have bandwidth less than $n$ collapse to a constant multiple of the identity. This allows to finish the proof of the bound on the subsequence of average velocities (\ref{main:pf:step1}). The proof of Theorem \ref{thm:main-step} is presented in the next section, Section \ref{sec:Induction}.

Consider a fixed period $n\in\Z^+$ of the periodic field $\mathcal{D}_n$. For $N$ being a sequence of multiples of $n$, i.e., $N= nk$ for $k\in\Z^+$,  we expand
\begin{equation}
\frac{1}{N}(\tau_N^{U_{t,n}}(X)-X)=\frac{1}{nk}\sum_{m=1}^{k}((U_{t,n}^n)^*)^{m}[X,U_{t,n}^n] (U_{t,n}^n)^{m-1},
\end{equation}
to see that
\begin{equation}\label{eq:norm}
\frac{1}{N}\|\tau^{U_{t,n}}_N(X)-X\| \leq \frac{1}{n}\|[X,U_{t,n}^n]\|=\frac1n\|\tau^{U_{t,n}}_n(X)-X\|.
\end{equation}
This is, the average velocities in the first $N=nk$ quantum walk steps, where $k\in\Z^+$,  is bounded by a bound for the average velocity in the first $n$ steps only.
So our main effort will be directed towards analyzing $U_{t,n}^n=(\mathcal{D}_n U_t)^n$ that represents the quantum walk in the presence of the local field $\mathcal{D}_n$ after $n$ steps.

\subsection{A quantum walk generated by $n$ unitary operators}\label{subsec:BC}

Note that  $U_t=V_t W_t$ in (\ref{U}) and (\ref{U-CMV}) can be written as a sum of  three operators that correspond to constant, linear, and quadratic terms of $t$, 
\begin{equation}\label{U:1BC}
U_t=r^2\idty_{\mathcal{H}}+r B t+ C t^2,
\end{equation}
where 
\begin{equation}\label{def:BC}
B:=T_{-1}^{\mathcal{H}}-T_1^{\mathcal{H}}, \quad C:=|\downarrow\rangle\langle\downarrow|\otimes T_{-1}+
|\uparrow\rangle\langle\uparrow|\otimes T_1.
\end{equation}
Here $T_{\pm k}^{\mathcal{H}}$ is the shift operator by $k$ on $\mathcal{H}$, defined as 
\begin{equation}\label{def:TH}
T_{\pm k}^{\mathcal{H}}=\sum_{j\in\Z}|\delta_{j\pm k}\rangle\langle\delta_j|, \, k\in\Z^+.
\end{equation}

For $j\in\Z$, define the generalized $U_{t}^{(j)}$ operators on $\mathcal{H}$
\begin{equation}\label{def:U0j}
U_{t}^{(j)}:=\mathcal{D}_n^j U_t (\mathcal{D}_n^j)^*=r^2\idty_{\mathcal{H}}+rt B^{(j)}+t^2 C^{(j)},
\end{equation}
where the
$B^{(j)}$ and $C^{(j)}$ operators are given by the formulas (A direct calculation using (\ref{def:BC}), (\ref{def:Dn}), and (\ref{def:Dn-2}))
\begin{eqnarray}\label{def:BC:j}
B^{(j)}&:=&\mathcal{D}_n^j B (\mathcal{D}_n^j)^*
=
\alpha^{-j}T_{-1}^{\mathcal{H}}-\alpha^j T_1^{\mathcal{H}} \nonumber\\
\quad C^{(j)}&:=&\mathcal{D}_n^j C (\mathcal{D}_n^j)^*= \alpha^{-2j}|\downarrow\rangle\langle \downarrow|\otimes T_{-1}+\alpha^{2j}|\uparrow\rangle\langle \uparrow|\otimes T_1. \label{def:gen-BC-U}
\end{eqnarray}
Here we slightly abuse notation and suppress the dependency on $n$ in the operators $U_t^{(j)}$, $B^{(j)}$, and $C^{(j)}$. This dependency is hidden within $\alpha$ which is an $n$-th root of the unity (\ref{def:alpha}). Note that
\begin{equation}\label{imp-1}
A^{(0)}=A,\text{ and } A^{(n+k)}=A^{(k)}, \text{ for all } k\in\Z, \text{ where } A\in\{B,C, U_t\}. 
\end{equation}
 Direct calculations using (\ref{def:gen-BC-U}) show the crucial facts
\begin{equation}\label{eq:important}
[B^{(j)},B^{(k)}]=[C^{(j)},C^{(k)}]=0, \text{ and }B^{(j)}C^{(k)}=(C^{(k)})^* B^{(j)} \text{ for all }j, k\in \Z.
\end{equation}
\begin{rem}
(\ref{imp-1}) and (\ref{eq:important}) are essential relations for the validity of our proof of the main result. In particular, Corollary \ref{cor:s-D} below is a consequence of these relations. It is important to stress here that  a local field (unitary and diagonal) $\mathcal{D}=\diag\{e^{i\theta_j}\}$, $j\in\Z$ satisfies 
\begin{equation}
[B,\mathcal{D}^k B (\mathcal{D}^k)^*]=[C,\mathcal{D}^k C (\mathcal{D}^k)^*]=0 \text{ for any }k\in\Z,\text{ and }\mathcal{D}^n=\idty_\mathcal{H}
\end{equation}
 if and only if each diagonal entry of $\mathcal{D}_n$ is an  $n$-th root of the unity  and  they are ordered as $e^{i(\theta_{j+1}-\theta_j)}=c$ for some $c\in\C$ for all $j\in\Z$, (compare with property (b) in the definition of $\mathcal{D}_n$). This can be checked by direct calculations. The requirement that $c=1$ corresponds to the trivial case of no field ($\mathcal{D}_n=\beta\idty_\mathcal{H}$, where $|\beta|=1$). $c\neq 1$ means that any consecutive $n$ diagonal entries of $\mathcal{D}_n$ are distinct, this is crucial for the validity of Lemma \ref{Com-with-D}.
\end{rem}
Next, use (\ref{def:U0j}) that $\mathcal{D}_n^\ell$ (where $\ell\in\Z^+$) satisfies the intertwining relation, $\mathcal{D}_n^\ell U_{t}=U_{t}^{(\ell)}\mathcal{D}_n^\ell$ iteratively as follows, to find
\begin{eqnarray}
U_{t,n}^n&=& (\mathcal{D}_n U_{t}) (\mathcal{D}_n U_{t})  \ldots  (\mathcal{D}_n U_{t} )  \quad (n \text{ times})\nonumber\\
&=&\mathcal{D}_n U_{t} U_{t}^{(1)}  \mathcal{D}_n^2 U_{t,n}\ldots  \mathcal{D}_n U_{t,n}\nonumber\\
&=&\mathcal{D}_n\left(\prod_{k=0}^{n-1}U_{t}^{(k)}\right) \mathcal{D}_n^{n-1},\label{eq:DUD}
\end{eqnarray}
where here and in the following we use the convention 
\begin{equation}\label{def:prod}
\prod_{m=a}^bK_m=
\begin{cases}
K_aK_{a+1}\ldots K_b & \text{ when }b\geq a\\
\idty_{\mathcal{H}} & \text{ when }b<a
\end{cases}.
\end{equation}

\eqref{eq:DUD} also reads as 
\begin{equation}
U_{t,n}^n=U_t^{(1)}U_t^{(2)}\ldots U_t^{(n)},
\end{equation}
noting that $U_t^{(j)}\neq U_t^{(\ell)}$ for all $j\neq \ell$ in $\{1,\ldots,n\}$. That is, 
the $N$-th quantum walk step, where $N=nk$, generated by the one-step unitary $U_{t,n}$ in (\ref{U}), is understood as the application of the block of operators $U_t^{(1)}U_t^{(2)}\ldots U_t^{(n)}$ $k$-times, i.e.,
\begin{equation}
U_{t,n}^{N}=\left(U_t^{(1)}U_t^{(2)}\ldots U_t^{(n)}\right)^k.
\end{equation}
This means that $U_{t,n}^n$ can be understood as the main building block of our quantum walk.

Substitute (\ref{eq:DUD}) in (\ref{eq:norm}) to obtain
\begin{equation}\label{eq:Loc-1}
\|[X,U_{t,n}^n]\|=\left\|\left[X, \mathcal{D}_n\left(\prod_{k=0}^{n-1}U_{t}^{(k)}\right) \mathcal{D}_n^{n-1}\right]\right\|=\left\|\left[X,\prod_{k=0}^{n-1}U_{t}^{(k)}\right]\right\|.
\end{equation}
In (\ref{eq:Loc-1}) we used Leibniz rule
and the fact that $X$ commutes with the unitary multiplication operators $\mathcal{D}_n$ and $\mathcal{D}_n^{n-1}$.

Recall that $\prod_{j=0}^{n-1}U_{t}^{(k)}$ is the product
\begin{equation}\label{eq:U-Product}
\prod_{j=0}^{n-1}U_{t}^{(k)}=\left(r^2\idty_{\mathcal{H}}+r B^{(0)}t+ C^{(0)}t^2\right) \ldots\left(r^2\idty_{\mathcal{H}}+r B^{(n-1)}t+C^{(n-1)}t^2\right)
\end{equation}
which can be regarded as a ``polynomial" in  $t$ of degree $2n$ whose coefficients are operators acting on $\mathcal{H}$. The full expansion has $3^n$ terms that we will group with respect to the corresponding power of $t$, as follows.

First note that in the product (\ref{eq:U-Product}), a product of a combination of $m_1$ operators $C^{(j)}$'s and $m_2$ operators $B^{(k)}$'s (with an increasing order of distinct superscripts in the range from $0$ to  $n-1$) contributes to the coefficient  of $t^{2m_1+m_2}$.
To distinguish such terms that are associated with a specific power of $t$ we introduce the operator $\mathcal{S}_{[0,n-1]}^{\ell, m}$ on $\mathcal{H}$ that is defined in the next subsection.

\subsection{The elementary symmetric polynomial of operators $\s_{[0,n-1]}^{\ell,m}$}\label{subsec:s}
We define the operator $\mathcal{S}_{[0,n-1]}^{\ell,m}$ on $\mathcal{H}$ to be the (symmetric) sum of products that corresponds to the power $m+\ell$ of $t$. In particular, each term in $\mathcal{S}_{[0,n-1]}^{\ell, m}$ is an ordered (in the superscripts)  product of a total of $\ell$ operators, $m$ of which are $C^{(j)}$'s and $\ell-m$ are $B^{(k)}$'s where 
\begin{equation}
j,k\in[0,n-1]:=\{0,1,\ldots,n-1\}.
\end{equation}
In general, $\s_{[0,n-1]}^{\ell,m}$ where $0\leq m\leq \ell\leq n$ is an operator on $\mathcal{H}$ given by the formula
\begin{equation}\label{eq:S-def}
\s_{[0,n-1]}^{\ell,m}=\sum_{
\tiny\begin{array}{c}
0\leq j_1<\ldots<j_\ell\leq n-1\\
\eta=(\eta_1,\ldots,\eta_\ell)\in\{0,1\}^\ell;\\
|\eta|=m
\end{array}
}\,\prod_{k=1}^\ell A_{\eta_k}^{(j_k)},
\end{equation}
where $|\eta| = \eta_1 + \ldots + \eta_{\ell}$ and $A$ is either a $C$ or a $B$ operator, i.e.,
 \begin{equation}
 A_{\eta}:= C^{\eta} B^{1-\eta},\ \eta\in\{0,1\} \text{ with the definition }C^{0} =B^{0}:=\idty_{\mathcal{H}},
 \end{equation} 
and we set $\s_{[0,n-1]}^{0,0}:=\idty_{\mathcal{H}}$, and note that $\s_{\{j\}}^{1,1}=C^{(j)}$ and $\s_{\{j\}}^{1,0}=B^{(j)}$. That is, (\ref{eq:S-def}) can be written as 
\begin{equation}
\s_{[0,n-1]}^{\ell,m}=\sum_{
\tiny\begin{array}{c}
\eta=(\eta_1,\ldots,\eta_\ell)\in\{0,1\}^\ell;\\
|\eta|=m
\end{array}
}\, \sum_{
\tiny\begin{array}{c}
0\leq j_1<\ldots<j_\ell\leq n-1
\end{array}
}\,\prod_{k=1}^\ell \s_{\{j_k\}}^{1,\eta_k}.
\end{equation}
We remark here that the object $\s_{[0,n-1]}^{\ell,m}$  takes the form of the sum of the so-called  \emph{elementary symmetric polynomials} of operators (see e.g., \cite{ESP-1, ESP-2}). Here the variables are the non-commuting operators $B^{{(j)}}$ and $C^{(k)}$'s.

\begin{ex}
The border case $\s_{[0,n-1]}^{\ell,\ell}$ is given as
\begin{equation}
\s_{[0,n-1]}^{\ell,\ell}=\sum_{0\leq j_1<j_2<j_3\leq n-1} C^{(j_1)} C^{(j_2)}\ldots C^{(j_\ell)}.
\end{equation}
We also have
\begin{equation}
\mathcal{S}_{[1,3]}^{2,1}= C^{(1)}B^{(2)}+ C^{(1)}B^{(3)}+ C^{(2)}B^{(3)} +B^{(1)}C^{(2)}+ B^{(1)}C^{(3)}+ B^{(2)}C^{(3)},
\end{equation}
and $\mathcal{S}_{[0,n-1]}^{3,1}$ is a sum of products of three operators, one $C$ and two $B$'s,
\begin{equation}
\mathcal{S}_{[0,n-1]}^{3,1}=\sum_{0\leq j_1<j_2<j_3\leq n-1} C^{(j_1)} B^{(j_2)} B^{(j_3)}+B^{(j_1)} C^{(j_2)} B^{(j_3)}+B^{(j_1)} B^{(j_2)} C^{(j_3)}.
\end{equation}
\end{ex}
Note that $\s_{[0,n-1]}^{\ell,m}$ has exactly
\begin{equation}\label{S:number}
\binom{n}{m}\binom{n-m}{\ell-m}=\binom{n}{\ell}\binom{\ell}{m}
\end{equation}
terms (products). This can be seen for example from the sum in (\ref{eq:S-def}), where one needs to place $\ell$ values of $j_k$'s within $n$ positions $\{0,1,\ldots, n-1\}$, then we place $m$ ones in $\ell$ positions in $(\eta_1,\ldots,\eta_\ell)$.

\subsection{Expanding in terms of  $\s_{[0,n-1]}^{\ell,m}$}\label{subsec:expand}

It is direct to see from the product (\ref{eq:U-Product}) that the constant (in $t$) term is $(r^2)^n\idty_{\mathcal{H}}$, and the linear term is
\begin{equation}
(r^2)^{n-1}r \mathcal{S}_{[0,n-1]}^{1,0} t=r^{2n-1} \mathcal{S}_{[0,n-1]}^{1,0} t.
\end{equation}
In general, the product (\ref{eq:U-Product}) can be written as (recall that $\mathcal{S}_{[0,n-1]}^{0,0}=\idty_{\mathcal{H}}$)
\begin{equation}
\prod_{j=0}^{n-1}U_{t}^{(j)}=\sum_{k=0}^{2n} r^{2n-k} t^k\sum_{\tiny\begin{array}{c}0\leq m\leq \ell\leq n\\  \ell+m=k\\
\end{array}
}\mathcal{S}_{[0,n-1]}^{\ell,m}.
\end{equation}
Then we write the commutator $[X, U_{t,n}^n]$ as the sum
\begin{equation}\label{eq:exact-long-formula}
\left[X,U_{t,n}^n\right]=
\sum_{k=0}^{2n}r^{2n-k}  t^k \sum_{\tiny\begin{array}{c}0\leq m\leq \ell\leq n\\  \ell+m=k\\
\end{array}
}\left[X, \mathcal{S}_{[0,n-1]}^{\ell,m}\right].
\end{equation}
The following main theorem shows that the first $n$ terms (from $k=0$ to $k=n-1$) in the sum (\ref{eq:exact-long-formula}) are zeros. In fact, we need to show only that $\s_{[0,n-1]}^{\ell,m}$ is a diagonal operator (a multiplication operator) whenever $\ell+m<n$. The following main technical theorem, proved in Section \ref{sec:Proof-of-thm}, shows that $\s_{[0,n-1]}^{\ell,m}$ is not only diagonal, but it is a constant multiple of the identity. 

\begin{thm}\label{thm:main-step}
For any $0\leq m\leq \ell\leq n$ with $\ell+m<n$, the sum $\s_{[0,n-1]}^{\ell,m}$ is a constant multiple of the identity. (Note that $\s_{[0,n-1]}^{0,0}=\idty_{\mathcal{H}}$.)
\end{thm}

Theorem \ref{thm:main-step} implies that $[X,\s_{[0,n-1]}^{\ell,m}]=0$ for all $0\leq m\leq \ell\leq n$ such that $\ell+m\in[0,n-1]$. Thus, we read from (\ref{eq:norm})  and (\ref{eq:exact-long-formula}) that
\begin{equation}\label{eq:hope}
\frac{1}{N}\|\tau^{U_{t,n}}_N(X)-X\|\leq \frac{1}{n}\|[X,U_{t,n}^n]\|= \frac1n\, t^n\left\|
L_n
\right\|
\end{equation}
where
\begin{equation}
L_n=\sum_{k=0}^{n}r^{n-k}  t^k \sum_{\tiny\begin{array}{c}0\leq m\leq \ell\leq n\\  \ell+m=n+k\\
\end{array}
}\left[ X,\mathcal{S}_{[0,n-1]}^{\ell,m}\right].
\end{equation}
\begin{rem}
In the following, we present an argument that shows that $\|L_n\|\leq 3n\, 4^{n-1}$, and this is the ``four" in the velocity bound $(4t)^n$ in Theorem \ref{thm:main}. We think that this bound is not optimal and that $\|L_n\|=\mathcal{O}(n)$. This would lead to the result of Conjecture \ref{conj} above, but we could not prove a better bound.
\end{rem}
First, note that
\begin{equation}\label{eq:Com-bound-0}
\|L_n\|\leq \sum_{\tiny\begin{array}{c}
\ell,m\in\Z\\
\frac{n}{2}\leq \ell\leq n\\
n-\ell\leq m\leq \ell
\end{array}} \left\|\left[X, \mathcal{S}_{[0,n-1]}^{\ell,m}\right]\right\|
\end{equation}
where we observe that
\begin{equation}
\{(\ell,m)|\ 0\leq m\leq \ell\leq n,\ \ell+m\in[n,2n]\}=\{(\ell,m)|\ \frac{n}{2}\leq \ell\leq n,\ n-\ell\leq m\leq \ell\},
\end{equation}
which is clear from Figure \ref{figure:1}, see the shaded region.

\begin{figure}[h]
\begin{center}
\includegraphics[width=4in]{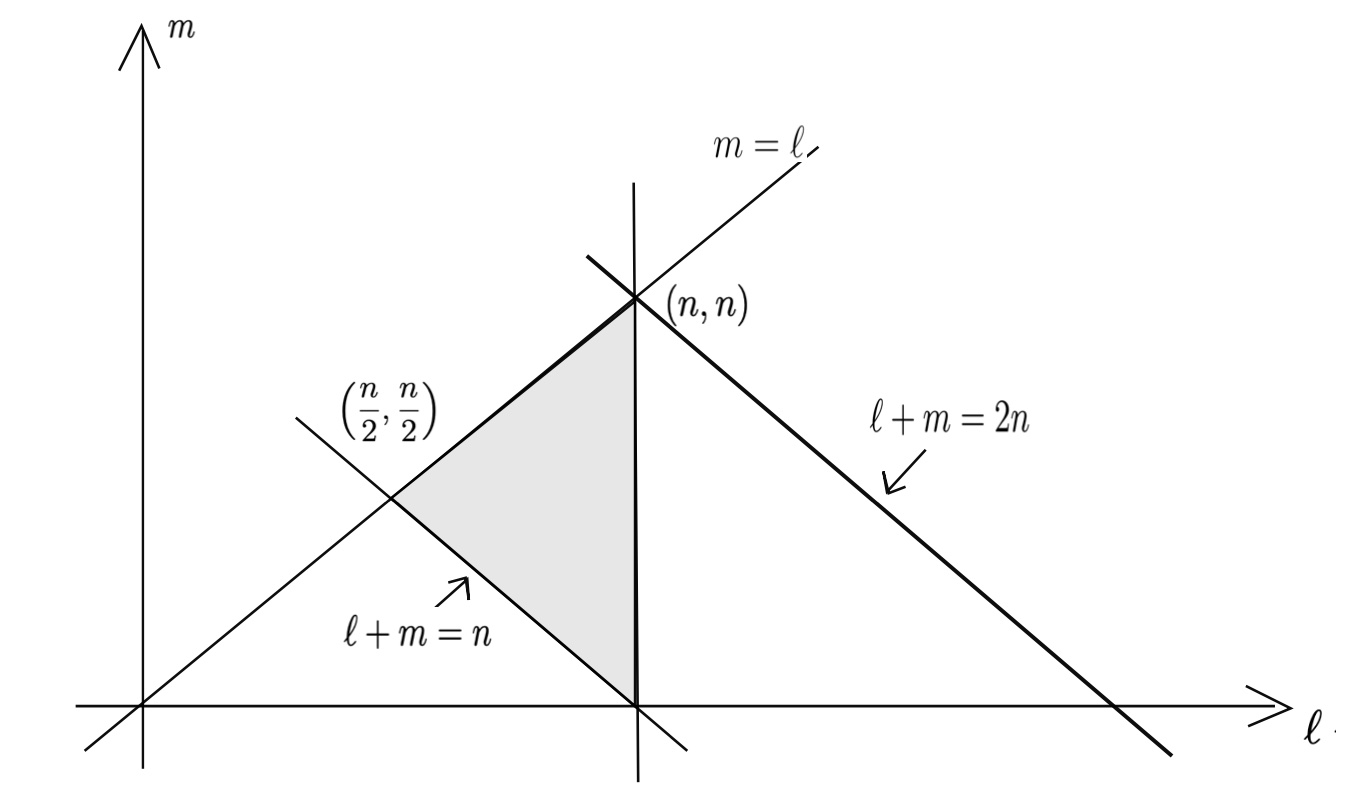}
\caption{The shaded region corresponds to the nonzero terms in (\ref{eq:exact-long-formula}) and the sum in (\ref{eq:Com-bound-0})}
\label{figure:1}
\end{center}
\end{figure}

We then use the following lemma that we prove at the end of this section.
\begin{lem}\label{lem:Com-bound}
For any $0\leq m\leq \ell\leq n$, we have the bound
\begin{equation}\label{eq:Com-bound}
\left\|\left[X, \mathcal{S}_{[0,n-1]}^{\ell,m}\right]\right\|\leq 2^{\ell-m}\ \ell \binom{n}{\ell}\binom{\ell}{m}.
\end{equation}
\end{lem}
Use the bound (\ref{eq:Com-bound}) in (\ref{eq:Com-bound-0}) to find
\begin{eqnarray}
\|L_n\| &\leq&\sum_{\ell=0}^n \ell 2^{\ell}\binom{n}{\ell}\sum_{m=0}^\ell \binom{\ell}{m}2^{-m}\nonumber\\
&=& 3n\, 4^{n-1}
\end{eqnarray}
where we used the binomial sum identity and the elementary fact
\begin{equation}
\sum_{\ell=0}^n \ell \binom{n}{\ell}x^{\ell-1} =\frac{d}{dx}\sum_{\ell=0}^n \binom{n}{\ell}x^{\ell} =n(x+1)^{n-1}.
\end{equation}
This finishes the proof of the \emph{Step 1}, and shows that
\begin{equation}
\frac{1}{N}\|\tau^{U_{t,n}}_N(X)-X\|\leq \frac{1}{n}\|[X,U_{t,n}^n]\|= \frac34\, (4t)^n,
\end{equation}
for any $N$ being a multiple of the period $n$.

It remains to prove Lemma \ref{lem:Com-bound}.
\begin{proof}[Proof of Lemma \ref{lem:Com-bound}]
First recall that for any $0\leq m\leq \ell\leq n$,
\begin{equation}\label{eq:S-def-recall}
[X,\s_{[0,n-1]}^{\ell,m}]=\sum_{
\tiny\begin{array}{c}
0\leq j_1<\ldots<j_\ell\leq n-1\\
\eta=(\eta_1,\ldots,\eta_\ell)\in\{0,1\}^\ell\\
|\eta|=m
\end{array}
}\,\,  \left[X,\prod_{k=1}^\ell A_{\eta_k}^{(j_k)}\right],\quad \text{ where }A_\eta=C^\eta B^{1-\eta}, \, \eta\in\{0,1\}.
\end{equation}
Observe that for any $j\in\Z$, $C^{(j)}=\mathcal{D}_n^j C (\mathcal{D}_n^{j})^*$ is unitary (recall that $C$ is an orthogonal operator, it is a permutation). Hence, $\|C^{(j)}\|=1$. It is also clear that $\|B^{(j)}\|=\|B\|\leq 2$ for any $j\in\Z$. That is
\begin{equation}\label{eq:nom-A}
\|A_{\eta_k}^{(j_k)}\|=\|A_{\eta_k}\|\ \begin{cases}
=1 &\text{ if } \eta_k=1\\
\leq 2 & \text{ if } \eta_k=0\\
\end{cases}.
\end{equation}
Next we use the fact (\ref{eq:XBC}) in Lemma \ref{lem:BC} below, that $\|[X,B]\|=\|[X,C]\|=1$, to see that
\begin{equation}
\|[X,A_{\eta_k}^{(j_k)}]\|=\|[X,\mathcal{D}_n^{j_k} A_{\eta_k}(\mathcal{D}_n^*)^{j_k}]\|=\|[X,A_{\eta_k}]\|= 1.
\end{equation}
Now, we expand the commutator in (\ref{eq:S-def-recall}) using Leibniz rule
\begin{equation}
\left[X,\prod_{k=1}^\ell A_{\eta_k}^{(j_k)}\right]=\sum_{r=0}^{\ell-1} \left(\prod_{k=1}^r A_{\eta_k}^{(j_k)}\right)[X,A_{\eta_{r+1}}^{(j_{r+1})}]\left(\prod_{k=r+2}^\ell A_{\eta_k}^{(j_k)}\right),
\end{equation}
where we recall our convention for the product of operators (\ref{def:prod}).

Take the norm while choosing  $\eta=(\eta_1,
\ldots,\eta_\ell)\in\{0,1\}^\ell$ with $|\eta|=m$ to obtain
\begin{eqnarray}
\left\|\left[X,\prod_{k=1}^\ell A_{\eta_k}^{(j_k)}\right]\right\|&\leq&
 \sum_{r=0}^{\ell-1} \|[X,A_{\eta_{r+1}}^{(j_{r+1})}]\|\prod_{k,\ k\neq r+1}\|A_{\eta_k}^{(j_k)}\|\nonumber\\
 &\leq&  \sum_{r=0}^{\ell-1} \prod_{k=1}^\ell\|A_{\eta_k}\|\nonumber \\
 &\leq& \ell\cdot 2^{\ell-m}
\end{eqnarray}
where we used (\ref{eq:nom-A}) in the last step with the fact that there are $(\ell-m)$ of $B$'s in the product.
We have then
\begin{eqnarray}
\left\|\left[X, \mathcal{S}_{[0,n-1]}^{\ell,m}\right]\right\|
&\leq&
\ell\ 2^{\ell-m} \sum_{
\tiny\begin{array}{c}
0\leq j_1<\ldots<j_\ell\leq n-1\\
\eta=(\eta_1,\ldots,\eta_\ell)\in\{0,1\}^\ell\\
|\eta|=m
\end{array}
} 1 \nonumber\\
&=&
2^{\ell-m}\ \ell \binom{n}{\ell}\binom{\ell}{m},
\end{eqnarray}
see (\ref{S:number}).
\end{proof}

\section{A proof of Theorem \ref{thm:main-step}}\label{sec:Induction}

This section includes the proof of the main technical result, Theorem \ref{thm:main-step}. This is done in the following two subsection. Section \ref{subsec:proof-1}  includes some preparatory results that will be crucial in the proof by induction in the Section \ref{subsec:proof-2}. In short, Lemma \ref{lem:idty} below adds to Lemma \ref{Com-with-D} and shows that if $\s_{[1,n-1]}^{\ell,m}$ commutes with the periodic $n$-local field $\mathcal{D}_n$  defined in (\ref{def:Dn}) for a certain values of $\ell$ and $m$ then $\s_{[1,n-1]}^{\ell,m}$ is not only diagonal (result of Lemma \ref{Com-with-D}) but it is a constant multiple of the identity. Lemma \ref{lem:rec-formulas} and Corollary \ref{cor:s-D} present some recursive formulas that will be essential in the inductive assumptions in Section \ref{subsec:proof-2}.

\subsection{Some properties of the symmetric polynomial $\s_{[0,n-1]}^{\ell,m}$}\label{subsec:proof-1}

In this section, we present some properties that we will use to prove Theorem \ref{thm:main-step}. 

The key step in the proof of Theorem \ref{thm:main-step} is to prove that $\s_{[0,n-1]}^{\ell,m}$ commutes with $\mathcal{D}_n$, and since $\s_{[0,n-1]}^{\ell,m}$ has bandwidth $<n$ when $\ell+m<n$, then this means that it is a diagonal matrix, see Lemma \ref{Com-with-D}. In fact, the following Lemma shows a stronger conclusion if $\s_{[0,n-1]}^{\ell,m}$ commutes with $\mathcal{D}_n$. Let us define \begin{equation}\label{def:Xi}
\Xi_n:=\{(\ell,m)\in\Z^2;\, 0\leq m\leq \ell \leq n \text{ such that }\ell+m<n\}.
\end{equation}
This is the set that corresponds to $\s_{[0,n-1]}^{\ell,m}$ with bandwidth $<n$.

\begin{lem}\label{lem:idty}
If $[\s_{[0,n-1]}^{\ell,m},\mathcal{D}_n]=0$ and $(\ell,m)\in\Xi_n$ then $\s_{[0,n-1]}^{\ell,m}$ is a constant multiple of the identity, that is, $\s_{[0,n-1]}^{\ell,m}$ commutes with any operator on $\mathcal{H}$.
\end{lem}
While it is needed only to show that $\s_{[0,n-1]}^{\ell,m}$ is a diagonal operator when $(\ell,m)\in\Xi_n$ to prove the main result, we will need the stronger result of Lemma \ref{lem:idty} in an inductive proof (in the inductive assumption) to show that $[\s_{[0,n-1]}^{\ell,m},\mathcal{D}_n]=0$ for all $(\ell, m)\in\Xi_n$ in Theorem \ref{thm:main-step}.
\begin{proof}
It is clear that $\s_{[0,n-1]}^{\ell,m}$ is a banded operator with bandwidth at most $2m+(\ell-m)=\ell+m$. If $[\s_{[0,n-1]}^{\ell,m},\mathcal{D}_n]=0$ and $\ell+m<n$ then Lemma \ref{Com-with-D} gives directly that all non-diagonal elements  $\langle\delta_j,\s_{[0,n-1]}^{\ell,m}\delta_k\rangle$, $j \not= k$ are zeros. 

This proves that $\s_{[0,n-1]}^{\ell,m}$ is diagonal whenever $(\ell, m)\in\Xi_n$ (given that $[\s_{[0,n-1]}^{\ell,m},\mathcal{D}_n]=0$). In the following we prove that 
\begin{equation}
\s_{[0,n-1]}^{\ell,m}=\begin{cases}
\gamma\idty_\mathcal{H} & \text{ if }\ell-m\text{ is even}\\
0 &  \text{ if }\ell-m\text{ is odd}
\end{cases}, \text{ where }\gamma\in\C.
\end{equation}

\noindent\underline{Case 1:} $\ell-m$ is odd.\\
Define the diagonal operators on $\mathcal{H}=\C^2\otimes \ell^2(\Z)$
\begin{equation}\label{def:a}
a^{\uparrow}:=|\uparrow\rangle\langle\uparrow|\otimes \idty_\Z \text{ and }
a^{\downarrow}:=|\downarrow\rangle\langle\downarrow|\otimes \idty_\Z 
\end{equation}
and observe that $a^\uparrow a^\downarrow=0$ and $a^\uparrow+a^\downarrow=\idty_{\mathcal{H}}$. A direct inspection yields that $a^\uparrow$ and $a^\downarrow$ commute with $C^{(j)}$, i.e., 
\begin{equation}
[a^\uparrow,C^{(j)}]=[a^{\downarrow}, C^{(j)}]=0,
\end{equation}
and they satisfy the intertwining relation
\begin{equation}
a^\uparrow B^{(j)}=B^{(j)}a^{\downarrow},  \text{ and }
 a^{\downarrow} B^{(j)} = B^{(j)} a^\uparrow.
\end{equation}
Next observe that since $\s_{[0,n-1]}^{\ell,m}$ with odd $\ell-m$ corresponds to  an odd number of $B$'s, then it is easy to see that
\begin{equation}
a^\uparrow \s_{[0,n-1]}^{\ell,m}=\s_{[0,n-1]}^{\ell,m} a^{\downarrow}=a^{\downarrow}\s_{[0,n-1]}^{\ell,m}
\end{equation}
where the last equality follows from our assumption that $\s_{[0,n-1]}^{\ell,m}$ is diagonal (for $(\ell, m)\in\Xi_n$). This shows directly that $\s_{[0,n-1]}^{\ell,m}=0$.

\noindent\underline{Case 2:} $\ell-m$ is even.
\begin{rem}
Observe that $T_{-2}^{\mathcal{H}}=\idty_2\otimes T_{-1}$ commutes with both $B^{(j)}$ and $C^{(k)}$ operators ($T_{-2}^{\mathcal{H}}$ is the shift operator on $\mathcal{H}$, see (\ref{def:TH})), i.e.,
$
[T_{-2}^{\mathcal{H}},B^{(j)}]=[T_{-2}^{\mathcal{H}},C^{(k)}]=0 \text{  for any }j,k\in\Z.
$
Hence we conclude that $
[T_{-2}^{\mathcal{H}}, \s_{[0,n-1]}^{\ell,m}]=0$,
meaning that $\langle \delta_j,\s_{[0,n-1]}^{\ell,m}\delta_j\rangle=\langle \delta_{j+2},\s_{[0,n-1]}^{\ell,m}\delta_{j+2}\rangle$ for all $j\in\Z$, which says that $\s_{[0,n-1]}^{\ell,m}=\gamma_1 a^\uparrow+\gamma_2 a^{\downarrow}$, for some $\gamma_1,\gamma_2\in\C$, here $a^\uparrow$ and $a^\downarrow$ are the diagonal operators defined in (\ref{def:a}). This short argument does not prove the Theorem, and instead, we need to show that $[\s_{[0,n-1]}^{\ell,m},T_{-1}^{\mathcal{H}}]=0$, which is much more involved.
\end{rem}

Our main goal is to prove that $\s_{[0,n-1]}^{\ell,m}$ commutes with $T_{-1}^{\mathcal{H}}$ noting that
\begin{equation}
[T_{-1}^{\mathcal{H}},B^{(j)}]=0, \text{ but }T_{-1}^{\mathcal{H}}C^{(k)}=(C^{(k)})^* T_{-1}^{\mathcal{H}}.
\end{equation}
$[\s_{[0,n-1]}^{\ell,m},T_{-1}^{\mathcal{H}}]=0$ implies that $\langle \delta_j,\s_{[0,n-1]}^{\ell,m}\delta_j\rangle=\langle \delta_{j+1},\s_{[0,n-1]}^{\ell,m}\delta_{j+1}\rangle$ for all $j\in\Z$, which is the desired result.

Recall that $\s_{[0,n-1]}^{\ell,m}$ can be written as 
\begin{equation}
\s_{[0,n-1]}^{\ell,m}=\sum_{
\tiny\begin{array}{c}
0\leq j_1<\ldots<j_\ell\leq n-1\\
\eta=(\eta_1,\ldots\eta_\ell)\in\{0,1\}^\ell;\\
|\eta|=m
\end{array}
}\, \prod_{k=1}^\ell A_{\eta_k}^{(j_k)},\text{ where }A_{\eta_k}:= C^{\eta_k} B^{1-\eta_k}.
\end{equation}
Then, in the case when there is an even number of $B$'s
\begin{equation}\label{eq:T-1-S}
T_{-1}\s_{[0,n-1]}^{\ell,m}=\hat\s_{[0,n-1]}^{\ell,m} T_{-1}, 
\end{equation}
where 
\begin{equation}\label{def:S-hat}
\hat\s_{[0,n-1]}^{\ell,m}:=\sum_{\tiny\begin{array}{c}
0\leq j_1<\ldots<j_\ell\leq n-1\\
\eta=(\eta_1,\ldots\eta_\ell)\in\{0,1\}^\ell;\\
|\eta|=m
\end{array}}\, \prod_{k=1}^\ell (A_{\eta_k}^{(j_k)})^*.
\end{equation}
In (\ref{def:S-hat}) we used the fact that $(B^{(j)})^*=-B^{(j)}$ for all $j\in\Z$. Observe that since $\s_{[0,n-1]}^{\ell,m}$ is given (assumed) to be a diagonal operator then 
\begin{equation}\label{eq:T-1-S-2}
T_{-1}\s_{[0,n-1]}^{\ell,m}=(\s_{[0,n-1]}^{\ell,m})^{(\uparrow)}T_{-1}
\end{equation}
where $(\cdot)^{(\uparrow)}$ denotes the ``up'' shift by one, i.e.,
\begin{equation}
(\s_{[0,n-1]}^{\ell,m})^{(\uparrow)}:=\sum_{j\in\Z} \langle\delta_j,\s_{[0,n-1]}^{\ell,m}\delta_j\rangle|\delta_{j-1}\rangle\langle\delta_{j-1}|.
\end{equation}
So, (\ref{eq:T-1-S}) and (\ref{eq:T-1-S-2}) show that $\hat\s_{[0,n-1]}^{\ell,m}=(\s_{[0,n-1]}^{\ell,m})^{(\uparrow)}$, and hence $\hat\s_{[0,n-1]}^{\ell,m}$ is also diagonal. Thus, it is enough to show that the diagonal entries of $\hat\s_{[0,n-1]}^{\ell,m}$ and $\s_{[0,n-1]}^{\ell,m}$ are equal. This follows from the following claim, then from the symmetry with respect to $\eta\in\{0,1\}^\ell$.
\begin{cl}
For every fixed $\eta=(\eta_1,\ldots,\eta_\ell)\in\{0,1\}^\ell$ and $k\in \Z$,
\begin{equation}\label{eq:diag-entries}
\sum_{0\leq j_1<\ldots<j_\ell\leq n-1} \left\langle \delta_k, A_{\eta_1}^{(j_1)}A_{\eta_2}^{(j_2)}\ldots A_{\eta_\ell}^{(j_\ell)}\delta_k\right\rangle= \sum_{0\leq j_1<\ldots<j_\ell\leq n-1} \left\langle \delta_k, (A_{\eta_\ell}^{(j_1)})^*(A_{\eta_{\ell-1}}^{(j_2)})^*\ldots (A_{\eta_1}^{(j_\ell)})^* \delta_k\right\rangle.
\end{equation}
\end{cl}
The following is a proof of the claim.
First, note that the diagonal elements in (\ref{eq:diag-entries}) are invariant under shifts in $j$'s, in the meaning that for any $c\in\Z$
\begin{eqnarray}\label{eq:shift-inv}
\left\langle  \delta_k,A_{\eta_1}^{(j_1+c)}A_{\eta_2}^{(j_2+c)}\ldots A_{\eta_\ell}^{(j_\ell+c)} \delta_k\right\rangle
&=&
\left\langle  \delta_k, \mathcal{D}_n^c A_{\eta_1}^{(j_1)}A_{\eta_2}^{(j_2)}\ldots A_{\eta_\ell}^{(j_\ell)}(\mathcal{D}_n^{c})^* \delta_k\right\rangle \nonumber\\
&=& \left\langle   \delta_k, A_{\eta_1}^{(j_1)}A_{\eta_2}^{(j_2)}\ldots A_{\eta_\ell}^{(j_\ell)} \delta_k\right\rangle
\end{eqnarray}
where we use the fact that $A^{(j+c)}=\mathcal{D}_n^c A^{(j)} (\mathcal{D}_n^{c})^*$ and $\mathcal{D}_n$ is the $n$-periodic field (unitary and diagonal), see (\ref{def:BC:j}).

Next, in the LHS of (\ref{eq:diag-entries}), consider the change of variables 
\[
(j_1,j_2,\ldots, j_\ell)\mapsto(-j_\ell,\ldots,-j_2,-j_1),
\]
 then use the fact $A^{(-j)}=\overline{A^{(j)}}$,
\begin{eqnarray}
\sum_{0\leq j_1<\ldots<j_\ell\leq n-1} \left\langle  \delta_k,\prod_{k=1}^\ell A_{\eta_k}^{(j_k)} \delta_k\right\rangle
&=&
\sum_{-(n-1)\leq j_1<\ldots<j_\ell\leq 0} \left\langle  \delta_k, A_{\eta_1}^{(-j_\ell)}A_{\eta_2}^{(-j_{\ell-1})}\ldots A_{\eta_\ell}^{(-j_1)} \delta_k\right\rangle \nonumber\\
&=&\sum_{-(n-1)\leq j_1<\ldots<j_\ell\leq 0} \overline{\left\langle  \delta_k, A_{\eta_1}^{(j_\ell)}A_{\eta_2}^{(j_{\ell-1})}\ldots A_{\eta_\ell}^{(j_1)} \delta_k\right\rangle} \nonumber\\
&=&
\sum_{0\leq j_1<\ldots<j_\ell\leq n-1} \overline{\left\langle \delta_k, A_{\eta_1}^{(j_\ell)}A_{\eta_2}^{(j_{\ell-1})}\ldots A_{\eta_\ell}^{(j_1)} \delta_k\right\rangle} \nonumber\\
&=&\sum_{0\leq j_1<\ldots<j_\ell\leq n-1}\left\langle \delta_k,(A_{\eta_\ell}^{(j_1)})^*(A_{\eta_{\ell-1}}^{(j_2)})^*\ldots (A_{\eta_1}^{(j_\ell)})^* \delta_k\right\rangle.
\end{eqnarray}
In the second-to-third step we used the  invariance under shifts in $j$'s (\ref{eq:shift-inv}). This marks the end of the proof of the claim and finishes Case 2 (when $\ell-m$ is even) in Lemma \ref{lem:idty}.
\end{proof}

We need to highlight some important recursive properties of $\s_{[0,n-1]}^{\ell,m}$. 
\begin{lem}\label{lem:rec-formulas}
For all $0\leq m\leq \ell\leq n$, we have the recursive formulas
\begin{equation}\label{eq:rec}
\s_{[0,n-1]}^{\ell,m}=\begin{cases}
B\s_{[1,n-1]}^{\ell-1,0}+\s_{[1,n-1]}^{\ell,0}\, \chi_{\ell\neq n} & \text{ if }m=0\\[0.2cm]
C\s_{[1,n-1]}^{\ell-1,m-1}+B\s_{[1,n-1]}^{\ell-1,m}+\s_{[1,n-1]}^{\ell,m} \, \chi_{\ell\neq n}& \text{ if } 1\leq m<\ell\\[0.2cm]
C\s_{[1,n-1]}^{\ell-1,\ell-1}+\s_{[1,n-1]}^{\ell,\ell}\, \chi_{\ell\neq n} & \text{ if } m=\ell\in[1,n]
\end{cases}
\end{equation}
and 
\begin{equation}\label{2}
 \mathcal{D}_n\s_{[0,n-1]}^{\ell,m} (\mathcal{D}_n)^*=
 \begin{cases}
 \s_{[1,n-1]}^{\ell-1,0}B+\s_{[1,n-1]}^{\ell,0}\, \chi_{\ell\neq n} & \text{ if } m=0\\[0.2cm]
 \s_{[1,n-1]}^{\ell-1,m-1}C+\s_{[1,n-1]}^{\ell-1,m}B+\s_{[1,n-1]}^{\ell,m}\, \chi_{\ell\neq n} & \text{ if }1\leq m<\ell\\[0.2cm]
  \s_{[1,n-1]}^{\ell-1,\ell-1}C+\s_{[1,n-1]}^{\ell,\ell} \, \chi_{\ell\neq n}& \text{ if }m=\ell\in[1,n]
 \end{cases}
 \end{equation}
 where $\chi_{\ell\neq n}= 1$ if $\ell\neq n$ and zero otherwise.
\end{lem}
Observe that  $\s_{[1,n-1]}^{\ell-1,0}$ and $\s_{[1,n-1]}^{\ell-1,\ell-1}$ are sums of products of only $B$, and only $C$ operators, respectively. Since $[B,B^{(j)}]=[C,C^{(j)}]=0$ for all $j\in\Z$, we observe that $[B,\s_{[1,n-1]}^{\ell-1,0}]=[C,\s_{[1,n-1]}^{\ell-1,\ell-1}]=0$. Compare  (\ref{eq:rec}) and (\ref{2}) to conclude the results of the following Corollary.

\begin{cor}\label{cor:s-D}
For any $0\leq \ell\leq n$, 
\begin{equation}\label{eq:s-m0-Com-D}
[\s_{[0,n-1]}^{\ell,0},\mathcal{D}_n]=0\text{  and  } [\s_{[0,n-1]}^{\ell,\ell},\mathcal{D}_n]=0
\end{equation}
and for general $1\leq m\leq \ell\leq n$
  \begin{equation}\label{eq:com-with-d}
 [\s_{[0,n-1]}^{\ell,m},\mathcal{D}_n]=0 \iff 
  [C, \s_{[1,n-1]}^{\ell-1,m-1}]+[B,\s_{[1,n-1]}^{\ell-1,m}]=0.
   \end{equation}
 \end{cor}
The strategy to see (\ref{eq:rec}) is that we can write $\s_{[0,n-1]}^{\ell,m}$ as a sum of three big sums: The first sum has all terms that start by $C^{(0)}$ (assuming that $m\geq 1$), the second sum includes the terms that start with $B^{(0)}$, and the last sum includes the remaining products. A similar argument applies for (\ref{2}). Here are the details.

\begin{proof}[Proof of Lemma \ref{lem:rec-formulas}]

We start by presenting the proof for the simple case when $m=0$, then the case $m\geq 1$ is a direct generalization.

Recall that
\begin{equation}\label{eq:s-ell-0}
\s_{[0,n-1]}^{\ell,0}=\sum_{0\leq j_1<j_2<\ldots<j_\ell\leq n-1}B^{(j_1)}B^{(j_2)}\ldots B^{(j_\ell)}.
\end{equation}
Then we can distinguish the terms that start with $B$ as follows
\begin{eqnarray}
\s_{[0,n-1]}^{\ell,0}&=&\sum_{1\leq j_2<j_3<\ldots<j_\ell\leq n-1}B^{(0)}\prod_{k=2}^{\ell}B^{(j_k)}
+\chi_{\ell\neq n}\, 
\sum_{1\leq j_1<j_2<\ldots<j_\ell\leq n-1}\prod_{k=1}^\ell B^{(j_k)} \nonumber\\
&=&B\s_{[1,n-1]}^{\ell-1,0}+\s_{[1,n-1]}^{\ell,0}\, \chi_{\ell\neq n}
\end{eqnarray}
which proves the first case of (\ref{eq:rec}). To see (\ref{2}) when $m=0$,
multiply $\s_{[0,n-1]}^{\ell,0}$ in (\ref{eq:s-ell-0}) from left by $\mathcal{D}_n$ and apply \eqref{def:gen-BC-U}, that $\mathcal{D}_n B^{(j)}=B^{(j+1)}\mathcal{D}_n$ to find 
\begin{eqnarray}
\mathcal{D}_n\s_{[0,n-1]}^{\ell,0}&=&
\sum_{0\leq j_1<j_2<\ldots<j_\ell\leq n-1}\prod_{k=1}^\ell B^{(j_k+1)}\mathcal{D}_n \nonumber\\
&=&
\sum_{1\leq j_1<j_2<\ldots<j_{\ell-1}\leq n-1}\left(\prod_{k=1}^{\ell-1} B^{(j_k)}\right)B\mathcal{D}_n+
\chi_{\ell\neq n}\, \sum_{1\leq j_1<j_2<\ldots<j_\ell\leq n-1}\prod_{k=1}^\ell B^{(j_k)}\mathcal{D}_n
 \nonumber\\
&=& (\s_{[1,n-1]}^{\ell-1,0}B+\chi_{\ell\neq n}\, \s_{[1,n-1]}^{\ell,0})\mathcal{D}_n
\end{eqnarray}
where we used the fact that $B^{(n)}=B^{(0)}$.

For general $m\geq 1$, we note that $\s_{[0,n-1]}^{\ell,m}$ can be written, by considering the cases of $j_1=0$ and $j_1\neq 0$, as the sum of the following terms.
\begin{eqnarray}\label{eq:S-rec-0}
\s_{[0,n-1]}^{\ell,m}
&=&
\sum_{
\tiny\begin{array}{c}
j_1=0,\ 1\leq j_2<\ldots<j_\ell\leq n-1\\
\eta=(\eta_1,\ldots\eta_\ell)\in\{0,1\}^\ell\\
|\eta|=m
\end{array}
}A_{\eta_1}^{(0)}\prod_{k=2}^\ell A_{\eta_k}^{(j_k)}
+\chi_{\ell\neq n}\, 
\s_{[1,n-1]}^{\ell,m}
\end{eqnarray}
and the first sum can also be split with respect to $\eta_1\in\{0,1\}$ as follows
\begin{eqnarray}
1^{\text{st}}\text{ term in (\ref{eq:S-rec-0})}&=&
\sum_{
\tiny\begin{array}{c}
j_1=0,\ 1\leq j_2<\ldots<j_\ell\leq n-1\\
\eta_1=1,\ \eta=(\eta_2,\ldots\eta_\ell)\in\{0,1\}^{\ell-1}\\
|\eta|=m-1
\end{array}
}\ldots
+
\sum_{
\tiny\begin{array}{c}
j_1=0,\ 1\leq j_2<\ldots<j_\ell\leq n-1\\
\eta_1=0,\ \eta=(\eta_2,\ldots\eta_\ell)\in\{0,1\}^{\ell-1}\\
|\eta|=m
\end{array}
}\ldots\nonumber\\
&=&
C \s_{[1,n-1]}^{\ell-1,m-1}+B\s_{[1,n-1]}^{\ell-1,m}.
\end{eqnarray}

We obtain
\begin{equation}
\s_{[0,n-1]}^{\ell,m}=C\s_{[1,n-1]}^{\ell-1,m-1}+B\s_{[1,n-1]}^{\ell-1,m}+\chi_{\ell\neq n}\, \s_{[1,n-1]}^{\ell,m}.
\end{equation}
Next we find $\mathcal{D}_n\s_{[0,n-1]}^{\ell,m}$ using $\mathcal{D}_n C^{(j)}=C^{(j+1)}\mathcal{D}_n$ and $\mathcal{D}_n B^{(j)}=B^{(j+1)}\mathcal{D}_n$ then take $B$ and $C$ as common factors from the ``tails" of products as follows. 
\begin{eqnarray}\label{eq:rec-D-S-0}
\mathcal{D}_n\s_{[0,n-1]}^{\ell,m}\mathcal{D}_n^*
&=&
\sum_{
\tiny\begin{array}{c}
0\leq j_1<\ldots<j_\ell\leq n-1\\
\eta=(\eta_1,\ldots\eta_\ell)\in\{0,1\}^\ell;\\
|\eta|=m
\end{array}
}\, \prod_{k=1}^\ell A_{\eta_k}^{(j_k+1)}
=
\sum_{
\tiny\begin{array}{c}
1\leq j_1<\ldots<j_\ell\leq n\\
\eta=(\eta_1,\ldots\eta_\ell)\in\{0,1\}^\ell;\\
|\eta|=m
\end{array}
}\, \prod_{k=1}^\ell A_{\eta_k}^{(j_k)} \nonumber\\
&=&
\chi_{\ell\neq n}\, \s_{[1,n-1]}^{\ell,m}
+
\sum_{
\tiny\begin{array}{c}
1\leq j_1<\ldots<j_{\ell-1}\leq n-1,\ j_{\ell}=n\\
\eta=(\eta_1,\ldots\eta_\ell)\in\{0,1\}^\ell;\\
|\eta|=m
\end{array}
}\, \left(\prod_{k=1}^{\ell-1} A_{\eta_k}^{(j_k)}\right)A_{\eta_\ell}^{(0)}.
\end{eqnarray}
Then we split the second term (the sum) as
\begin{eqnarray}
2^{\text{nd}} \text{ term in (\ref{eq:rec-D-S-0})} &=&\sum_{
\tiny\begin{array}{c}
1\leq j_1<\ldots<j_{\ell-1}\leq n-1,\ j_{\ell}=n\\
\eta=(\eta_1,\ldots\eta_{\ell-1})\in\{0,1\}^{\ell-1},\ \eta_\ell=0\\
|\eta|=m
\end{array}
}\ldots
+
\sum_{
\tiny\begin{array}{c}
1\leq j_1<\ldots<j_{\ell-1}\leq n-1,\ j_{\ell}=n\\
\eta=(\eta_1,\ldots\eta_{\ell-1})\in\{0,1\}^{\ell-1},\ \eta_\ell=1\\
|\eta|=m-1
\end{array}
}\ldots \nonumber\\
&=& \s_{[1,n-1]}^{\ell-1,m}B+\s_{[1,n-1]}^{\ell-1,m-1}C.
\end{eqnarray}
Thus, we proved that 
 \begin{equation}
 \mathcal{D}_n\s_{[0,n-1]}^{\ell,m} \mathcal{D}_n^* =\s_{[1,n-1]}^{\ell-1,m-1}C+\s_{[1,n-1]}^{\ell-1,m}B+\s_{[1,n-1]}^{\ell,m}\, \chi_{\ell\neq n}.
 \end{equation}
\end{proof}

\subsection{An induction argument on $\ell$ and $m$ }\label{sec:Proof-of-thm}\label{subsec:proof-2}

Now we are ready to present the proof of Theorem \ref{thm:main-step}. That is, we will show that $\s_{[0,n-1]}^{\ell,m}$ is a constant multiple of the identity for all $(\ell,m)\in \Xi_n$, where we recall here (for the reader's convenience) that
\begin{equation}
\Xi_n:=\{(\ell,m)\in\Z^2;\, 0\leq m\leq \ell \leq n \text{ such that }\ell+m<n\}.
\end{equation}

\begin{proof}[Proof of Thm \ref{thm:main-step}]
We use induction on $m$ and $\ell$ to show that $[\s_{[0,n-1]}^{\ell,m},\mathcal{D}_n]=0$ for all $(\ell,m)\in\Xi_n$, then Lemma \ref{lem:idty} gives directly, the desired result.

We use an induction argument for $(\ell,m)\in\Xi_n$ in three steps, see Figure \ref{Induction}, where we prove the following
\begin{itemize}
\item \emph{Step 1}: $[\s_{[0,n-1]}^{\ell,0},\mathcal{D}_n]=0$ for all $1\leq \ell<n$, and $[\s_{[0,n-1]}^{\ell,\ell},\mathcal{D}_n]=0$ for all $1\leq \ell<n/2$.
\item \emph{Step 2}: $[\s_{[0,n-1]}^{\ell,1},\mathcal{D}_n]=0$ for all $2\leq\ell <n-1$.
\item \emph{Step 3}: Here we use the inductive argument, if $[\s_{[0,n-1]}^{\ell,k},\mathcal{D}_n]=[\s_{[0,n-1]}^{\ell,k-1},\mathcal{D}_n]=0$ then $[\s_{[0,n-1]}^{\ell+1,k},\mathcal{D}_n]=0$ for all $2\leq k\leq \ell$ such that $\ell+k<n$.
\end{itemize}

\begin{figure}
\begin{center}
\includegraphics[width=4in]{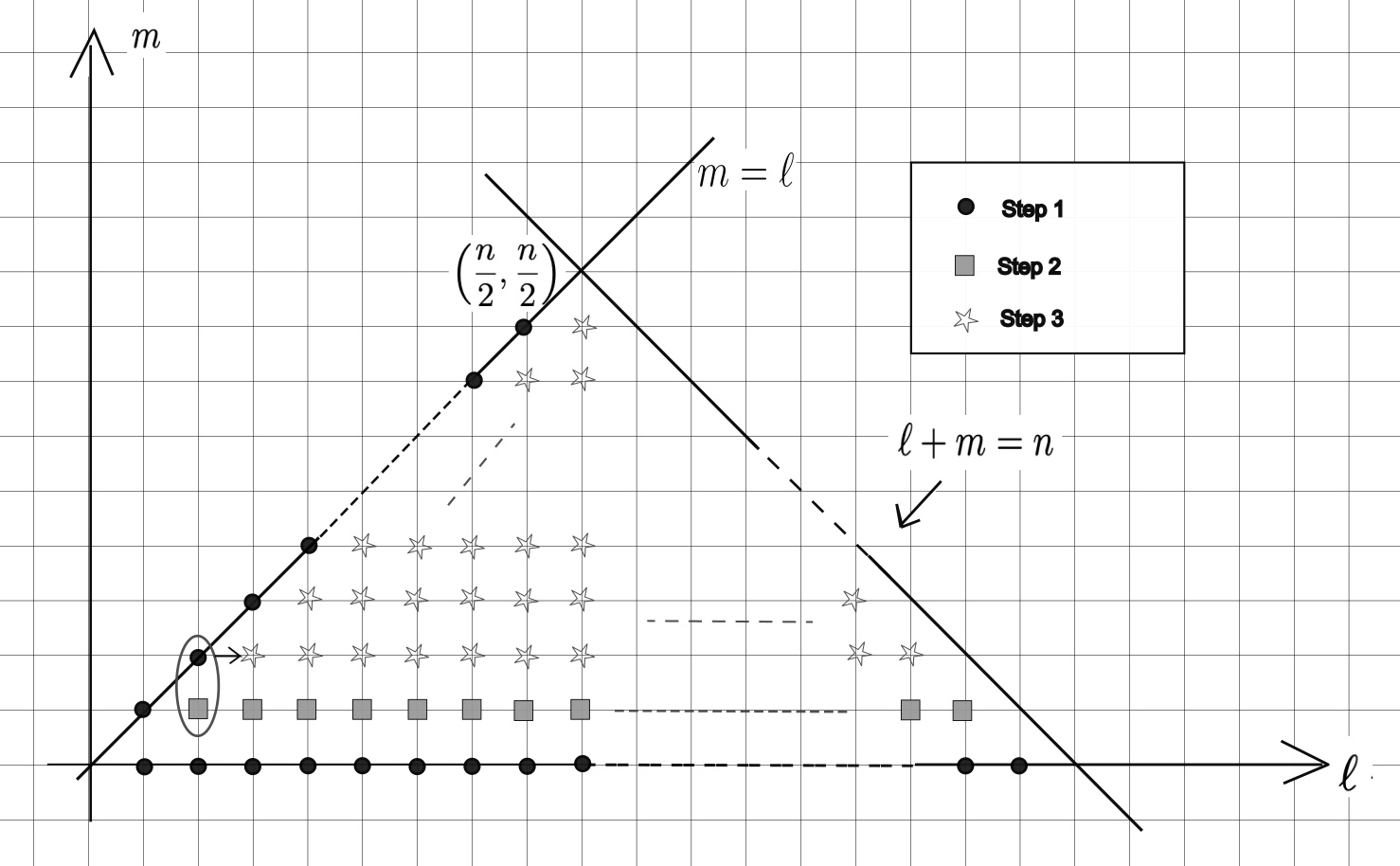}
\caption{Region $\Xi_n\subset\Z^2$ and the steps of proof by induction on $\ell$ and $m$.}
\label{Induction}
\end{center}
\end{figure}
\noindent\underline{Proof of \emph{Step 1}:} This is result (\ref{eq:s-m0-Com-D}) in Corollary \ref{cor:s-D}.
Lemma \ref{lem:idty} implies that $\s_{[0,n-1]}^{\ell,0}$ is constant multiple of the identity for all $\ell<n$, and that $\s_{[0,n-1]}^{\ell,\ell}$ is constant multiple of the identity for all $2\leq 2\ell<n$.



 \noindent\underline{Proof of \emph{Step 2}:} We use induction on $1\leq \ell\leq n-2$. The initial step that $[\s_{[0,n-1]}^{1,1},\mathcal{D}_n]=0$ is proven in \emph{Step 1}.
 For the inductive step, assume  that $[\s_{[0,n-1]}^{\ell,1},\mathcal{D}_n]=0$ for a fixed $1\leq \ell <n-2$ (and hence it is a constant multiple of the identity, see Lemma \ref{lem:idty}). That is, we assume that (see Corollary \ref{cor:s-D})
 \begin{equation}\label{step3-assu}
 [C, \s_{[1,n-1]}^{\ell-1,0}]+[B, \s_{[1,n-1]}^{\ell-1,1}]=0
 \end{equation}
 and we want to show that $[\s_{[0,n-1]}^{\ell+1,1},\mathcal{D}_n]=0$.
Set $m=1$ in (\ref{eq:com-with-d}) in Corollary \ref{cor:s-D} to see that  
 \begin{equation}\label{step3-Result}
 [\s_{[0,n-1]}^{\ell+1,1},\mathcal{D}_n]=0 \iff [C, \s_{[1,n-1]}^{\ell,0}]+[B,\s_{[1,n-1]}^{\ell,1}]=0.
 \end{equation}
So, it is enough to show that
\begin{equation}
[C, \s_{[1,n-1]}^{\ell,0}]+[B,\s_{[1,n-1]}^{\ell,1}]=0.
\end{equation}
Note that Lemma \ref{lem:idty} together with \emph{Step 1} and the inductive assumption give that both operators $\s_{[0,n-1]}^{\ell,0}$ and  $\s_{[0,n-1]}^{\ell,1}$ are constant multiples of the identity, and recall from \eqref{eq:rec} that 
\begin{equation}\label{Step2:0}
\s_{[1,n-1]}^{\ell,0}=\s_{[0,n-1]}^{\ell,0}-B \s_{[1,n-1]}^{\ell-1,0} \text{ and }\s_{[1,n-1]}^{\ell,1}=\s_{[0,n-1]}^{\ell,1}-C\s_{[1,n-1]}^{\ell-1,0}-B\s_{[1,n-1]}^{\ell-1,1},
\end{equation}
and that $\s_{[0,n-1]}^{\ell,0}$  is a constant multiple of the identity, to see that
\begin{equation}\label{Step2:eq1}
[C, \s_{[1,n-1]}^{\ell,0}]=-[C, B\s_{[1,n-1]}^{\ell-1,0}].
\end{equation}
Similarly, since $\s_{[0,n-1]}^{\ell,1}$ is a constant multiple of the identity, (\ref{Step2:0}) gives
\begin{equation}\label{Step2:eq2}
[B,\s_{[1,n-1]}^{\ell,1}] =-[B,C\s_{[1,n-1]}^{\ell-1,0}+B\s_{[1,n-1]}^{\ell-1,1}].
\end{equation}

Then add equations (\ref{Step2:eq1}) and (\ref{Step2:eq2}), use Leibniz rule $[A,BC]=B[A,C]+[A,B]C$, and the fact that  $\s_{[1,n-1]}^{\ell-1,0}$ is a sum of products of the $B^{(j)}$ operators only, and hence it commutes with $B$, to see that
\begin{eqnarray}
[C, \s_{[1,n-1]}^{\ell,0}]+[B,\s_{[1,n-1]}^{\ell,1}]&=&-B\left([C, \s_{[1,n-1]}^{\ell-1,0}]+[B, \s_{[1,n-1]}^{\ell-1,1}]\right)-([C,B]+[B,C]) \s_{[1,n-1]}^{\ell-1,0} \nonumber \\
&=&0
\end{eqnarray}
where we used the induction assumption (\ref{step3-assu}).

\noindent\underline{Proof of \emph{Step 3}:} Assume that for any fixed $\ell$ and $m$ such that  $2\leq m\leq \ell$ and $\ell+m<n$, we have $[\s_{[0,n-1]}^{\ell,m},\mathcal{D}_n]=[\s_{[0,n-1]}^{\ell,m-1},\mathcal{D} _n]=0$, i.e., we assume that
\begin{eqnarray}\label{eq:step4}
[C, \s_{[1,n-1]}^{\ell-1,m-1}]+[B,\s_{[1,n-1]}^{\ell-1,m}]&=&0\nonumber \\[0em]
[C, \s_{[1,n-1]}^{\ell-1,m-2}]+[B,\s_{[1,n-1]}^{\ell-1,m-1}]&=&0
\end{eqnarray}
and we need to prove that 
\begin{equation}\label{step4-goal}
[\s_{[0,n-1]}^{\ell+1,m},\mathcal{D}_n]=0, \text{ that is }[C, \s_{[1,n-1]}^{\ell,m-1}]+[B,\s_{[1,n-1]}^{\ell,m}]=0.
\end{equation}
We proceed as in the previous step. Lemma \ref{lem:idty} together with \emph{Step 1} and the inductive assumption give that the operators $\s_{[0,n-1]}^{\ell,m}$ and  $\s_{[0,n-1]}^{\ell,m-1}$ are constant multiples of the identity. That is, by \eqref{eq:rec},
\begin{eqnarray}
[C, \s_{[1,n-1]}^{\ell,m-1}]&=& -[C,C\s_{[1,n-1]}^{\ell-1,m-2}+B\s_{[1,n-1]}^{\ell-1,m-1}]\\[0em]
[B,\s_{[1,n-1]}^{\ell,m}]&=& -[B,C\s_{[1,n-1]}^{\ell-1,m-1}+B\s_{[1,n-1]}^{\ell-1,m}].
\end{eqnarray}
Add these two formulas to see that $[C, \s_{[1,n-1]}^{\ell,m-1}]+[B,\s_{[1,n-1]}^{\ell,m}]$ can be written as
\begin{eqnarray}
&=& - C\left([C,\s_{[1,n-1]}^{\ell-1,m-2}]+[B,\s_{[1,n-1]}^{\ell-1,m-1}]\right)-\left([C,B]+[B,C]\right)\s_{[1,n-1]}^{\ell-1,m-1}+\nonumber\\
&&\hspace{4cm}-B\left([C,\s_{[1,n-1]}^{\ell-1,m-1}]+[B,\s_{[1,n-1]}^{\ell-1,m}]\right)\nonumber\\
&=& 0
\end{eqnarray}
where we used the induction assumptions (\ref{eq:step4}). This proves \eqref{step4-goal}, and finishes the inductive argument.
\end{proof}

\appendix
\section{Auxiliary Lemmas}

In this appendix we present some Lemmas that we use in the main body of this work.
\begin{lem}\label{Com-with-D}
Let $D=(\diag\{d_1,d_2,\ldots, d_n\})^{\oplus \Z}$ be a periodic operator of (fundamental) period  $n$ on $\ell^2(\Z)$ such  that $\{d_j\in\C,\, j=1,\ldots, n\}$ is a set of distinct complex numbers.
If a banded operator $K$  on   $\ell^2(\Z)$, with bandwidth less than $n$, commutes with $D$ then $K$ is a multiplication operator, i.e., $K$ is diagonal.
\end{lem}
\begin{proof}
Without loss of generality, we may assume that 
\begin{equation}
D=\sum_{j\in\Z} d_j |j\rangle\langle j|, \text{ where }d_{j+n}=d_j \text{ for all }j\in\Z. 
\end{equation}

If $K$ commutes with $D$, then
$
\langle j | K D | k\rangle=\langle j | D K | k\rangle
$
for all $j,k\in\Z$. So
\begin{equation}
d_k\langle j | K | k \rangle=d_j \langle j | K | k\rangle.
\end{equation}
From the construction of the $n$-periodic operator $D$, we have $d_j\neq d_k$ whenever $0<|j-k|<n$. Thus
$\langle j | K | k\rangle=0$ for all $j,k\in\Z$ such that $0<|j-k|<n$, i.e., $K$ is diagonal.
\end{proof}

\begin{lem}\label{lem:BC}
For the operators $B$ and $C$ defined (\ref{def:BC}) and the position operator $X$ in (\ref{def:X}) we have the following
\begin{equation}\label{eq:XBC}
\|[X,B]\|= \|[X,C]\|= 1, \text{ and }\|[X, r Bt+Ct^2]\|=t.
\end{equation}
\end{lem}
\begin{proof}
Note that $B=T_{-1}^{\mathcal{H}}-T_1^{\mathcal{H}}$ an be written as
\begin{equation}
B=\begin{pmatrix}0 & 1\\-1& 0\end{pmatrix}
\otimes \idty_\Z+|\downarrow\rangle\langle\uparrow| \otimes T_{-1}- |\uparrow\rangle\langle\downarrow|\otimes T_1
\end{equation}
and recall that
\begin{equation}
C=|\downarrow\rangle\langle\downarrow|\otimes T_{-1}+
|\uparrow\rangle\langle\uparrow|\otimes T_1.
\end{equation}
Then a direct calculation shows that
\begin{eqnarray}\label{X:BC}
[X,B]&=&|\downarrow\rangle\langle\uparrow|\otimes \sum_{j\in\Z}\eta_j^- |j-1\rangle\langle j|-
|\uparrow\rangle\langle\downarrow|\otimes \sum_{j\in\Z}\eta_j^+ |j+1\rangle\langle j| \nonumber \\[0cm]
[X,C] &=&
|\uparrow\rangle\langle\uparrow| \otimes \sum_{j\in\Z}\eta_j^+ |j+1\rangle\langle j |+
|\downarrow\rangle\langle\downarrow|\otimes \sum_{j\in\Z}\eta_j^- |j-1\rangle\langle j|
\end{eqnarray}
where
\begin{equation}
\eta_j^{\pm}:=|j\pm1|-|j|\leq |\eta_j^{\pm}|= 1 \text{ for all }j\in\Z.
\end{equation}
Then it is direct to see that, noting that $(\eta_j^\pm)^2=1$ for all $j\in\Z$
\begin{eqnarray}
[X,B]([X,B])^*&=&|\downarrow\rangle\langle\downarrow| \otimes \sum_{j\in\Z}(\eta_{j+1}^-)^2 |j\rangle\langle j |+
|\uparrow\rangle\langle\uparrow|\otimes \sum_{j\in\Z}(\eta_{j-1}^+)^2 |j\rangle\langle j|
=\idty_{\mathcal{H}}.
\end{eqnarray}
Similarly, we have $[X,C]([X,C])^*= \idty_{\mathcal{H}}$, 
showing that
\begin{equation}
\|X,B\|=\|[X,C]\|=1.
\end{equation}
Use the formulas in (\ref{X:BC}) and $(\eta_j^\pm)^2=1$ for all $j\in\Z$ to find directly that
\begin{equation}
[X, r t B+ t^2 C]([X,rt B+ t^2 C])^*=(r^2t^2+t^4)\idty_{\mathcal{H}}
\end{equation}
giving that
\begin{equation}
\|[X,rt B+ t^2 C]\|=\sqrt{r^2t^2+t^4}=t.
\end{equation}
\end{proof}

\section{A Proof of Lemma \ref{lem:v<t}}\label{App:proof:v<t}

Since $U_{t,\mathcal{D}}$ is unitary, we simplify
\begin{eqnarray}
\|\tau^{U_{t,\mathcal{D}}}_N(X)-X\|=\|(U_{t,\mathcal{D}}^*)^N (X U_{t,\mathcal{D}}^N-U_{t,\mathcal{D}}^N X)\|=\|[X,U_{t,\mathcal{D}}^N]\|.
\end{eqnarray}
Then we use Leibniz rule to expand the commutator
\begin{equation}
[X,U_{t,\mathcal{D}}^N]=\sum_{k=0}^{N-1}  U_{t,\mathcal{D}}^k[X,U_{t,\mathcal{D}}]U_{t,\mathcal{D}}^{N-k-1}
\end{equation}
and obtain
\begin{equation}\label{eq:linear:1}
\frac1N\|\tau^{U_{t,\mathcal{D}}}_N (X)-X\|=\frac1N \|[X,U_{t,\mathcal{D}}^N]\|\leq
\frac1N\sum_{k=0}^{N-1}\|[X,U_{t,\mathcal{D}}]\|= \|[X,U_{t,\mathcal{D}}]\|.
\end{equation}
Recall that the local field $\mathcal{D}$ is diagonal and unitary, thus
\begin{equation}\label{eq:linear:2}
\|[X,U_{t,\mathcal{D}}]\|=\|[X,U_t]\| 
=
\|[X, r^2\idty_{\mathcal{H}}+ rB t+C t^2]\| 
=\|[X, rB t+C t^2]\| 
= t,
\end{equation}
where we used Lemma \ref{lem:BC}, and (\ref{U:1BC}) that
\begin{equation}
U_t=r^2\idty_{\mathcal{H}}+r B t+C t^2
\end{equation}
where $B$ and $C$ are given in (\ref{def:BC}). This proves inequality (\ref{eq:v<t}).

To show that $v_{U_1,\mathcal{D}}=1$, we observe first that for any local field $\mathcal{D}$,
\begin{equation}\label{v1D-1}
U_{1,\mathcal{D}}=\mathcal{D}C=|\downarrow\rangle\langle\downarrow|\otimes \sum_j \xi_{2j-2} |j-1\rangle\langle j|+
|\uparrow\rangle\langle\uparrow|\otimes \sum_j \xi_{2j+1} |j+1\rangle\langle j|,
\end{equation}
where $\xi_j=\langle j|\mathcal{D}|j\rangle$ for all $j\in\Z$ and recall that $|\xi_j|=1$ for all $j$. This shows directly that
\begin{eqnarray}
U_{1,\mathcal{D}}^N&=&|\downarrow\rangle\langle\downarrow|\otimes \sum_j \left(\prod_{k=1, k \text{ is even}}^{2N-1}\xi_{2j-k}\right) |j-N\rangle\langle j|+\nonumber\\
&&\hspace{3cm}
+|\uparrow\rangle\langle\uparrow|\otimes \sum_j \left(\prod_{k=1, k \text{ is odd}}^{N}\xi_{2j+k}\right) |j+N\rangle\langle j|.
\end{eqnarray}
With the initial state $\Psi_0=\varphi \otimes |0\rangle$ where $\varphi\in\C^2$ and $\|\varphi\|=1$, we obtain (the unit vector)
\begin{equation}
U_{1,\mathcal{D}}^N\Psi_0=\langle\downarrow|\varphi\rangle\left(\prod_{k=1, k \text{ is even}}^{2N-1}\xi_{2j-k}\right)|\downarrow\rangle\otimes |-N\rangle+\langle\uparrow|\varphi\rangle\left(\prod_{k=1, k \text{ is odd}}^{N}\xi_{2j+k}\right)|\uparrow\rangle \otimes |N\rangle
\end{equation}
then $XU_{1,\mathcal{D}}^N\Psi_0=N U_{1,\mathcal{D}}^N\Psi_0$, meaning hat
$
\frac1N \left\langle X\right\rangle_{U_{1,\mathcal{D}}^N  \Psi_0}=1$.
Hence $v_{1,\mathcal{D}}=1$ for any field $\mathcal{D}$ and
\begin{equation}\label{v1D-last}
1=v_{1,\mathcal{D}}\leq\limsup_{N\rightarrow\infty}\frac1N \|\tau^{U_{t,\mathcal{D}}}_N(X)-X\| \leq 1
\end{equation}
see (\ref{eq:v<t}).




\begin{thebibliography}{99}



\bibitem{QD-1} Y. Aharonov, L. Davidovich, N. Zagury, 
\emph{Quantum random walks}, Phys. Rev. A  \textbf{48},
1687-1690 (1993).

\bibitem{QW-Loc1} A. Ahlbrecht, V. B. Scholz, A.H. Werner,
\emph{Disordered quantum walks in one lattice dimension},
J. Math. Phys. \textbf{52}, 102201  (2011).

\bibitem{QWVelocity} A. Ahlbrecht, H. Vogts, A. H. Werner, R. F. Werner, \emph{Asymptotic evolution of quantum walks with random coin}, J. Math. Phys. \textbf{52}, 042201 (2011).


\bibitem{QDM-1} J. Asch, P. Duclos, P. Exner, \emph{Stability of driven systems with growing gaps, quantum rings, and Wannier ladders}, J. Stat. Phys. \textbf{92}, 1053-1070 (1998).

\bibitem{AK} J. Asch, A. Knauf, \emph{Motion in periodic potentials}, Nonlinearity \textbf{11}, 175 (1998).


\bibitem{BHJ} O. Bourget, J. Howland, A. Joye, \emph{Spectral analysis of unitary band matrices}, Commun. Math. Phys. \textbf{234}, 191-227 (2003).

\bibitem{BB1988} G. Blatter, D. Browne, \emph{Zener tunneling and localization in small conducting rings},
Phys. Rev. B \textbf{37}, 3856 (1988).

\bibitem{CMV1}  M. J. Cantero, L. Moral, L. Vel\'{a}zquez, \emph{Five-diagonal matrices and zeros of orthogonal polynomials on the unit circle}, Linear Algebra Appl. \textbf{362}, 29-56  (2003).


\bibitem{App-secureQM} C. M. Chandrashekar, Th. Busch, \emph{Localized quantum walks as secured quantum memory}, EPL \textbf{110}, 10005 (2015).


\bibitem{App-QGates1} A. M. Childs, D. Gosset, Z. Webb, \emph{Universal computation by multiparticle quantum walk}, Science \textbf{339}, 791-794 (2013).

\bibitem{Phs-PQW4} Chou Chung-I, Ho Choon-Lin, \emph{Localization and recurrence of a quantum walk in a periodic potential on a line}, Chinese Phys. B \textbf{23}, 110302 (2014)


\bibitem{CMV-Math2} D. Damanik, J. Erickson, J. Fillman, G. Hinkle, A. Vu, \emph{Quantum intermittency for sparse CMV matrices with an application to quantum walks on the half-line}, J. Approx. Theory \textbf{208}, 59-84 (2016).

\bibitem{QWVelocity-Damanic} D. Damanik, J. Fillman, D. C. Ong, \emph{spreading estimates for quantum walks on the integer lattice via power-law bounds on transfer matrices},
 J. Math. Pures Appl. \textbf{105}, 293-341  (2016).

\bibitem{DLY} D. Damanik, M. Lukic, W. Yessen, \emph{Quantum dynamics of periodic and limit-periodic Jacobi and block Jacobi matrices with applications to some quantum many body problems},    
Commun. Math. Phys. \textbf{337}, 1535-1561 (2015).

\bibitem{Phs-PQW1} B. Danaci et al., \emph{Disorder-free localization in quantum walks},
Phys. Rev. A \textbf{103}, 022416 (2021).

\bibitem{Phs-PQW3} C. Di Franco, M. Paternostro, \emph{Localizationlike effect in two-dimensional alternate quantum walks with periodic coin operations},
Phys. Rev. A \textbf{91}, 012328 (2015).


\bibitem{Exp-2} C. Esposito  et al., \emph{Quantum walks of two correlated photons in a 2D synthetic lattice}, npj Quantum Inf. \textbf{8}, 34 (2022).


\bibitem{CMV-Math1} L. Fang, D. Damanik, S. Guo, \emph{Generic spectral results for CMV matrices with dynamically defined Verblunsky coefficients},  J. Funct.  Anal. \textbf{279}, 108803 (2020).




\bibitem{HJ11} E. Hamza, A. Joye, \emph{Correlated Markov quantum walks}, Ann. Henri Poincar\'{e} \textbf{13}, 1767-1805 (2012).

\bibitem{HJ14} E. Hamza, A. Joye, \emph{Spectral transition for random quantum walks on trees}, Commun. Math. Phys. \textbf{326}, 415-439 (2014). 

\bibitem{HJS9}  E. Hamza, A. Joye, G. Stolz, \emph{Dynamical localization for unitary Anderson models}, Math Phys. Anal. Geom. \textbf{12}, 381 (2009).

\bibitem{HJS6} E. Hamza, A. Joye, G. Stolz, \emph{Localization for random unitary operators}, Lett. Math. Phys. \textbf{75}, 255-272 (2006).

\bibitem{HS7} E. Hamza, G. Stolz, \emph{Lyapunov exponents for unitary Anderson models}, J. Math. Phys. \textbf{48}, 043301 (2007).



\bibitem{J11} A. Joye, \emph{Random unitary models and their localization properties}, in Entropy and the Quantum II, Contemporary Mathematics, Vol. 552 (American Mathematical Society, Providence, RI, 2011) pp. 117-134.  
 
\bibitem{J4} A. Joye, \emph{Density of states and Thouless formula for random unitary band matrices},
Ann. Henri Poincar\'{e} \textbf{5}, 347-379 (2004).

\bibitem{J5} A. Joye, \emph{Fractional moment estimates for random unitary band matrices}, Lett. Math. Phys. \textbf{72}, 51-64 (2005).

\bibitem{JM10} A. Joye, M. Merkli, \emph{Dynamical localization of quantum walks in random environments}, J. Stat. Phys. \textbf{140}, 1025-1053 (2010).

\bibitem{Exp-3} M. Karski et al., \emph{Quantum Walk in Position Space with Single Optically Trapped Atoms}, Science \textbf{325}, 174-177 (2009).

\bibitem{TopOrd-1} T. Kitagawa, E. Berg, M. Rudner, E. Demler,
\emph{Topological characterization of periodically driven quantum systems},
Phys. Rev. B \textbf{82}, 235114 (2010).

\bibitem{TopOrd3} T. Kitagawa  et al., \emph{Observation of topologically protected bound states in photonic quantum walks},
Nat. Commun. \textbf{3}, 882 (2012).

\bibitem{Kono09} N. Konno, \emph{One-dimensional discrete-time quantum walks on random environments},
Quantum Inf. Process \textbf{8}, 387399 (2009).

\bibitem{QWLectures9} N. Konno, \emph{Quantum Walks}, in ``Quantum Potential Theory'', Franz, Sch\'{u}rmann Edts, Lecture Notes in Mathematics, 1954, 309-452, (2009).

\bibitem{Phy-Localization1} Y. Lahini  et al., \emph{Anderson localization and nonlinearity in one-dimensional disordered photonic lattices}, Phys. Rev. Lett. \textbf{100}, 013906 (2008).

\bibitem{Exp-5} Y. Lahini, G. R. Steinbrecher, A. D. Bookatz, D. Englund, \emph{Quantum logic using correlated one-dimensional quantum walks}, npj
Quant. Inf. \textbf{4}, 2 (2018).



\bibitem{App-Info-Transport} M. Leonetti  et al., \emph{Secure information transport by transverse localization of light}, Sci. Rep. \textbf{6}, 29918 (2016).

\bibitem{CMV-Math3} L. Li, D. Damanik, Q. Zhou, \emph{Cantor spectrum for CMV matrices with almost periodic Verblunsky coefficients}, J. Funct.  Anal. \textbf{283}, 109709 (2022).

\bibitem{App-QGates2} N. B. Lovett, S. Cooper, M. Everitt, M. Trevers, V. Kendon, \emph{Universal quantum computation using the discrete-time quantum walk}, Phys. Rev. A \textbf{81}, 042330 (2010).

\bibitem{QD-2} D. Lenstra, W. van Haeringen, \emph{Elastic scattering in a normal-metal loop causing resistive electronic behavior}, Phys. Rev. Lett. \textbf{57}, 1623-1626 (1986).




\bibitem{ESP-1} I. G. Macdonald, \emph{Symmetric Functions and Hall Polynomials (2nd ed.)}. Oxford: Clarendon Press, 1995.

\bibitem{Phy-Localization2} L. Martin  et al., \emph{Anderson localization in optical waveguide arrays with off-diagonal coupling disorder}, Opt. Express \textbf{19}, 13636-13646 (2011).

\bibitem{Alg-3} A. Montanaro, \emph{Quantum algorithms: An overview}, npj Quant. Inf. \textbf{2}, 15023 (2016).

\bibitem{Phs-QW5}  D. T. Nguyen, D. A. Nolan, N. F. Borrelli, \emph{Localized quantum walks in quasi-periodic Fibonacci arrays of waveguides}, Opt. Express \textbf{27}, 886-898 (2019).



\bibitem{Phs-PQW2} D. T. Nguyen et al., \emph{Quantum Walks in Periodic and Quasi-periodic Fibonacci Fibers}, Sci. Rep. \textbf{10}, 7156 (2020). 

\bibitem{Exp-4} C. Noh, D. G. Angelakis, \emph{Quantum simulations and many-body physics with light}, Rep. Prog. Phys. \textbf{80}, 016401 (2017).


\bibitem{QDM-2} C. R. de Oliveira, M. S. Simsen, \emph{A Floquet operator with purely point spectrum
and energy instability}, Ann. H. Poincar\'{e} \textbf{7}, 1255-1277 (2008).


\bibitem{Alg-2} R. Portugal, \emph{Quantum walks and search algorithms}, Springer New York, 2013. 






\bibitem{QD-3} J.W. Ryu, G. Hur, S. W. Kim, \emph{Quantum localization in open chaotic systems},
Phys. Rev. E Stat. Nonlin. Soft Matter Phys. \textbf{78}, 037201 (2008).



\bibitem{App-QS1} L. Sansoni, F. Sciarrino, G. Vallone, P. Mataloni, A. Crespi, \emph{Two-particle bosonic-Fermionic quantum walk via integrated photonics}, Phys. Rev. Lett. \textbf{108}, 010502 (2012).

\bibitem{Exp-1} B. Sephton et al., \emph{A versatile quantum walk resonator with bright classical light}, PLoS ONE \textbf{14}, e0214891 (2019).

\bibitem{Alg-1} N. Shenvi, J. Kempe, K. B. Whaley, \emph{A quantum random walk search algorithm}, Phys. Rev. A \emph{67}, 052307 (2003).

\bibitem{CMV-Simon} B. Simon, \emph{CMV matrices: Five years after}, J. Comput. Appl. Math \textbf{208}, 120-154 (2006).


\bibitem{App-QS2} J. Spring  et al., \emph{Boson sampling on a photonic chip}, Science \textbf{339}, 798-801 (2013).

\bibitem{ESP-2} R. P. Stanley, \emph{Enumerative combinatorics}, Vol. 2. Cambridge: Cambridge University Press, 1999.





\bibitem{TopOrd-2} J. Wu, W.W. Zhang, B. C. Sanders, \emph{Topological quantum walks: Theory and experiments}, Front. of Phys. \emph{14}, 61301 (2019).

\bibitem{TopOrd4} L. Xiao  et al.,
\emph{Observation of topological edge states in parity-time-symmetric quantum walks}, Nat. Phys. \textbf{13}, 1117-1123 (2017).









\end{thebibliography}
\end{document}